\documentclass{article}
\usepackage{arxiv}
\usepackage{amsmath} 
\usepackage{amsthm}
\usepackage[utf8]{inputenc} 
\usepackage[T1]{fontenc}    
\usepackage{hyperref}       
\usepackage{url}            
\usepackage{booktabs}       
\usepackage{amsfonts}       
\usepackage{nicefrac}       
\usepackage{microtype}      
\usepackage{lipsum}		
\usepackage{graphicx}
\usepackage{enumitem}
\usepackage{wrapfig}
\usepackage{doi}
\usepackage[font=small]{caption}
\usepackage{listings}
\usepackage{xcolor}
\usepackage{titletoc}
\usepackage{placeins} 

\usepackage[
  backend=biber,
  style=numeric,
  sorting=none
]{biblatex}
\definecolor{codegreen}{rgb}{0,0.6,0}
\definecolor{codegray}{rgb}{0.5,0.5,0.5}
\definecolor{codepurple}{rgb}{0.58,0,0.82}
\definecolor{backcolour}{rgb}{0.95,0.95,0.92}

\lstdefinestyle{mystyle}{
  backgroundcolor=\color{backcolour}, commentstyle=\color{codegreen},
  keywordstyle=\color{magenta},
  numberstyle=\tiny\color{codegray},
  stringstyle=\color{codepurple},
  basicstyle=\ttfamily\footnotesize,
  breakatwhitespace=false,         
  breaklines=true,                 
  captionpos=b,                    
  keepspaces=true,                 
  numbers=left,                    
  numbersep=5pt,                  
  showspaces=false,                
  showstringspaces=false,
  showtabs=false,                  
  tabsize=2
}

\theoremstyle{definition}
\newtheorem{definition}{Definition}[section]

\newtheorem{conjecture}{Conjecture}[section]

\newtheorem{theorem}{Theorem}[section]

\title{Generalized saddle-node ghosts and their composite structures in dynamical systems}

\author{ \href{https://orcid.org/0000-0002-4893-1774}{\includegraphics[scale=0.06]{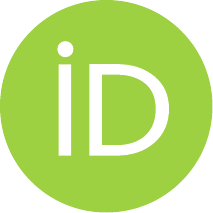}\hspace{1mm}Daniel~Koch}\thanks{To whom correspondence should be addressed.} \\
	Department of Pharmacology and Therapeutics\\
	University of Manitoba\\
    \&\\
    Institute of Cardiovascular Sciences\\
    St. Boniface Hospital, Albrechtsen Research Centre\\
	Winnipeg, MB, Canada \\
	\texttt{daniel.koch@umanitoba.ca} \\
	\And
	\href{https://orcid.org/0000-0003-4545-2687}{\includegraphics[scale=0.06]{orcid.pdf}\hspace{1mm}Akhilesh P.~Nandan} \\
	European Molecular Biology Laboratory \\
	  Barcelona, Spain \\
	\texttt{akhilesh.padmanabhan@embl.es} \\
}

\renewcommand{\shorttitle}{Generalized saddle-node ghosts and their composite structures}

\hypersetup{
pdftitle={Generalized saddle-node ghosts and their composite structures in dynamical systems},
pdfsubject={q-bio.NC, q-bio.QM},
pdfauthor={Daniel~Koch, Akhilesh P.~Nandan},
pdfkeywords={Ghost attractors, Bottlenecks, Ghost channels, Ghost cycles, Metastability, Transient dynamics, Saddle-node bifurcation, SNIC bifurcation, Dynamical Systems, Python package},
}

\begin{document}
\maketitle

\begin{abstract}
	The study of dynamical systems has long focused on the characterization of their asymptotic dynamics such as fixed points, limit cycles and other types of attractors and how these invariant sets change their properties as systems parameters change. More recently, however, the importance of transient dynamics, especially of long transients and sequential transitions between them, has been increasingly recognized in various fields including ecology, neuroscience, and cell biology. Among several possible origins of long transients, ghost attractors have received particular attention due to interesting dynamical properties in non-autonomous settings, new theoretical developments, and an increasing number of systems that empirically exhibit dynamics consistent with ghost attractors. Despite this growing interest in transient dynamics generally and ghost attractors in particular, there are significantly fewer theoretical concepts and software tools available to researchers to classify and characterize their underlying mechanisms compared to asymptotic dynamics. 
    To address this gap, we generalize saddle-nodes to account for higher-dimensional center manifolds and provide a definition for their ghost attractors. We then introduce algorithms to specifically identify and characterize ghost attractors and their composite structures such as ghost channels and ghost cycles and show how these concepts and algorithms can be used to gain new insights into the transient dynamics of a wide range of system models focusing on living systems, allowing, e.g., to describe bifurcations of ghosts. The algorithms are implemented in Python and available as {\tt PyGhostID}, a user-friendly open-source software package.
    
\end{abstract}


\begin{refsection}[references_clean.bib]
\section{Introduction}

Much of dynamical systems research focuses on understanding a system of interest through the lens of asymptotic dynamics, i.e., by studying how invariant sets such as fixed points, limit cycles and other types of attractors change as systems parameters are varied. Examples of this approach include tipping of climate- and ecosystems into undesirable states via passing through saddle-node bifurcations \cite{Scheffer1993-hn,Ibelings2007,Schmitt2019,Wunderling2021,Lohmann2024-pk}, the birth of the heart beat as emergence of a limit cycle \cite{Jia2023-wo} and its degeneration into cardiac arrhythmia as transition to chaos \cite{Garfinkel1992,Weiss1999,Sato2010}, or cellular decisions as the selection of distinct attractors of multi-stable molecular regulatory networks \cite{Xiong2003-ah,Ozbudak2004-vv,Byrne2016-kw}. Although studying attractors and their transitions between them has 
indisputably generated crucial insights into innumerable systems across the sciences, there are limitations to this approach. Indeed, many real-world systems often do not exhibit true asymptotic dynamics, e.g., due to frequent influences from their surroundings that render them non-autonomous. Some forms of cardiac arrhythmia, for instance, may best be described by transients that occur during transitions between steady states triggered by a sudden input signal rather than as asymptotic dynamics \cite{Xie2014}. In other cases, systems seemingly reside in a stable state before suddenly shifting towards a qualitatively different regime even in absence of any change of parameters or external inputs to the system, challenging the conceptualization of such dynamics as attractors. The importance of transient dynamics, especially of long transients - or \textit{metastable} states - and sequential transitions between them, is thus increasingly recognized in various fields, including ecology \cite{Hastings_2018,Morozov2024}, neuroscience \cite{Rossi_2025,Hancock2024} and cell biology \cite{Farjami_2021,Kelsh2021,Koch_2024}. 

Several mechanisms have been put forth to explain the emergence of long transients (see, e.g., \cite{Hastings_2018} and \cite{Rossi_2025} for reviews). Historically the first and mathematically best understood one among these is the explicit separation of timescales, resulting in what are called today slow-fast or multi-timescale systems \cite{Tikhonov1952,Verhulst2005,Kuehn_2015,Bertram2017}. In such systems, one or more processes take place on a timescale that typically differs by an order of magnitude or more and is often accounted for with a small parameter $\epsilon>0$, scaling the rate of change for one of the system variables. Using geometric singular perturbation theory, the slow dynamics in such systems can be studied as the flow on the so-called \textit{critical manifold} in the limit of $\epsilon=0$, the results of which can be shown to hold for the \textit{slow manifold} at $\epsilon>0$ due to Fenichel theory \cite{Fenichel_1971,Fenichel_1979}. While not all slow-fast or multi-timescale systems give rise to long transients, the cyclical visitation of slow manifolds in many oscillatory systems such as the FitzHugh-Nagumo model or the van der Pol oscillator naturally gives rise to alternating sequences of long transients. A different mechanism frequently invoked in ecology and neuroscience is the visitation of saddles by the system's phase space trajectory, either in the form of saddle "crawl-bys" \cite{Hastings_2018} or heteroclinic channels/cycles \cite{May1975}, causing the system to slow down in the vicinity of the visited saddle fixed point. Heteroclinic channels, cycles and networks are particularly suited for modeling sequential transitions between multiple transient states and have found applications, e.g., in neuroscience \cite{Rabinovich2000,Afraimovich2004} and robotic control \cite{Horchler2015}.

Another notable mechanism, and the focus of this paper, is the slowing down due to a \textit{bottleneck} or \textit{ghost} \cite{Strogatz1989,Strogatz_2018} when a system is close to a saddle-node (SN) bifurcation. From a geometrical perspective, the slow dynamics results from a plateau in the potential landscape of the system which is a remnant of a former attractor close-by in parameter space \cite{Stanoev_2020,Koch_2024}. Ghosts are a well-known theoretical mechanism for long transients underpinned by a growing body of empirical evidence in many systems including shallow-water lakes \cite{VanGeest2007}, neuronal dynamics \cite{Izhikevich_2006}, and cells \cite{Stanoev2018,RodrguezMaroto2025,Mercadal2026}. Being close to bifurcation, ghosts are particularly interesting in the context of non-autonomous systems, i.e., when system parameters or inputs are explicitly time-varying, often resulting in quite different responses compared to pure slow-fast systems \cite{Koch2025}. This feature has also been shown to improve information processing capabilities in biological systems \cite{Stanoev_2020,Moor2023,Choi_2024}. More recently, we showed that multiple ghosts can be connected to each other to form larger composite structures, which we termed \textit{ghost channels} and \textit{ghost cycles} \cite{Koch_2024}, providing an alternative description of sequential long transients. The trapping times in the presence of noise and the response of these objects to periodic forcing are quite different from heteroclinic structures \cite{Koch_2024,Rabinovich2006} and slow-fast systems \cite{Koch2025}, respectively.

Despite this variety of mechanisms, the theoretical and the computational tools available to researchers for characterizing transient dynamics are sparse compared to the tools available for asymptotic dynamics. While the lack of theoretical tools has sparked several promising new research avenues \cite{Liu2022,Heggerud2024,Zheng2025}, the lack of computational tools persists. While many powerful and widely used software solutions exist, e.g., for the characterization of attractors and their bifurcations \cite{doedel1981,Ermentrout2002,Dhooge2003,veltz:hal-02902346,Datseris2023}, the few algorithms and software solutions available to characterize transient dynamics (e.g., for numerical approximation of slow-manifolds via Computational Singular Perturbation \cite{Galassi2022}) are often rather domain specific or buried in sparsely documented research code. Thus, almost no widely applicable, user-friendly plug-and-play solutions for the characterization of transient dynamics in arbitrary dynamical systems exist to date, requiring researchers to resort to time-consuming manual analysis. 

In this article, we address this gap for one of the possible mechanisms underlying long transients and sequential transitions between them: ghost states and their composite structures. By giving a novel generalized definition of ghost states, we provide a new theoretical framework that allows to better describe ghosts as phase space objects and phenomena associated with ghosts, many of which have not been reported before. Based on this framework, we introduce a series of novel algorithms for identifying and characterizing ghost states and their composite structures (e.g., ghost channels \cite{Koch_2024}) along system trajectories, enabling us to obtain new insights into the transient dynamics, e.g., of neuronal and gene regulatory network models and to discover novel transient dynamics phenomena such as bifurcations between ghosts.

The remainder of this article is organized as follows. In \autoref{subsec:framework}, we review the basic concepts required to understand ghost states and their properties before giving a new formal definition that captures these features and generalize it to higher dimensional ghosts. Following the outline of the basic framework, in \autoref{subsec:algorithm} we describe the main algorithm of this paper used to identify and characterize ghost states from trajectories and how it is embedded in the PyGhostID package. In \autoref{subsec:validation}, we validate the algorithm, showing that GhostID specifically identifies ghost states but not saddles or slow-fast dynamics as origins of long transients. In \autoref{subsec:applications}, we show how our concepts and algorithms can be used to study various dynamical systems models with a focus on living systems. We conclude in \autoref{sec:discussion} by discussing current limitations of our approach, implications of the findings presented here, and avenues for future research.

\section{Results}\label{sec:results}
\subsection{Theoretical framework}\label{subsec:framework}

Before considering the non-asymptotic dynamics of ghost states, we briefly review some basic concepts about how the flow of a dynamical system is organized by fixed points (equilibria) as some of the most basic objects that govern a system's asymptotic dynamics. From a geometric perspective, fixed points are important because they direct the flow along defined directions and thereby determine a trajectory's fate. Given an autonomous dynamical system $\dot{x} = f(x,\rho)$ with parameters $\rho$, a fixed point is simply defined as any $x^*$ with $f(x^*,\rho)=0$. Different types of fixed points are classified by how the system's dynamics is organized in their vicinity. \textit{Stable} fixed points (also called attractors or sinks) attract all trajectories from their neighborhood, whereas \textit{unstable} fixed points (also called repellers or sources) direct trajectories to move away from the fixed point. As a simple example, consider the following system for $\mu <0$:
\begin{equation}
\dot{x} = f(x,\mu)=\mu + x^2. \label{eq:1}
\end{equation}
Visualizing the dynamics by plotting $\dot{x}$ vs $x$, we immediately see that the system has two fixed points, one at $x_s^* = -\sqrt{\mu}$ and one at $x_u^* = \sqrt{\mu}$ (\autoref{fig:fig1}a). Inferring the direction of flow by considering the value of $\dot{x}$ in the vicinity of $x_s^*$ (black arrowheads in \autoref{fig:fig1}a), we can see that trajectories will move towards $x_s^*$, hence $x_s^*$ is a stable fixed point. By the same kind of reasoning, we find that $x_u^*$ is unstable. In higher dimensions, sinks and repellers exist, too, but are complemented by saddles, defined to have one repelling and one attracting direction on the plane, or $k$, $0<k<n$, repelling and $n-k$ attracting directions in $n$ dimensions, respectively (\autoref{fig:fig1}b).

The notion of stability and the classification of fixed points can be made precise by linear stability analysis, in which a system is linearized around the fixed point and classified by the eigenvalues $\lambda_1,\dots,\lambda_n$ of the Jacobian matrix $D_x f(x^*)$ at the fixed point: a fixed point $x^*$ is a repeller if $\forall i:\textnormal{re}(\lambda_i)>0$, a sink if $\forall i:\textnormal{re}(\lambda_i)<0$ or a saddle if there are $k$ eigenvalues with $\textnormal{re}(\lambda_i)>0$ and $n-k$ eigenvalues with $\textnormal{re}(\lambda_i)<0$, respectively. Fixed points with $\forall i:\textnormal{re}(\lambda_i)\neq0$ are called \emph{hyperbolic}. In this study, however, we will be primarily concerned with fixed points that have at least one zero eigenvalue, i.e., with \emph{non-hyperbolic} fixed points. Non-hyperbolic fixed points commonly arise at bifurcations and are quite different from their hyperbolic counterparts: their stability cannot easily be determined via linear stability analysis and they are typically quite fragile (structurally unstable) as the topology of the phase portrait is altered by arbitrarily small perturbations to the vector field. While non-hyperbolic fixed points come in different varieties, we will focus next on \emph{saddle-node} fixed points (not to be confused with hyperbolic saddles) and their ghosts as organizing centers of a system's non-asymptotic dynamics. 

To gain intuition about saddle-nodes, we return to \autoref{eq:1}. As we increase $\mu$, $x_s^*$ and $x_u^*$ move closer and closer together until they merge into a single fixed-point at $\mu=0$ (\autoref{fig:fig1}c), which is called a one-dimensional saddle-node - the namesake of the SN-bifurcation. The eigenvalue of the Jacobian of a one-dimensional first-order ODE is equivalent to the first derivative evaluated at the fixed point, i.e., in this case $\lambda=f'(0,0)=0$, confirming that our saddle-node is non-hyperbolic.
Since $f(x,0)>0$ for both $x<0$ and $x>0$, all points on the negative half-line will approach $x^*=0$ as $t\to\infty$, whereas all points on the positive half-line will move away from it, such saddle-nodes are sometimes called \emph{semi-stable}. In this sense, the saddle-node inherited the negative half-line as a basin of attraction from the stable fixed point, and a repelling positive half-line from the unstable fixed point. Moving on to higher dimensions, we can see that various types of saddle-nodes can be constructed by pairing one or more semi-stable directions with stable and unstable directions that correspond to a system's center, stable and unstable eigendirections, respectively (\autoref{fig:fig1}d). This picture makes it furthermore clear that in higher dimensions, saddle-nodes exist whose eigenspace is identical with their center subspace and that many saddle-nodes have unstable directions, repelling all trajectories that do not begin exactly on a stable or center manifold. To describe these saddle-nodes formally, we begin by generalizing the definition of a saddle-node equilibrium of type $k$ \cite{Amaral2010,Amaral2012} to account for semi-simple eigenvalues zero:

\begin{figure}
	\centering
	\includegraphics[width=\textwidth]{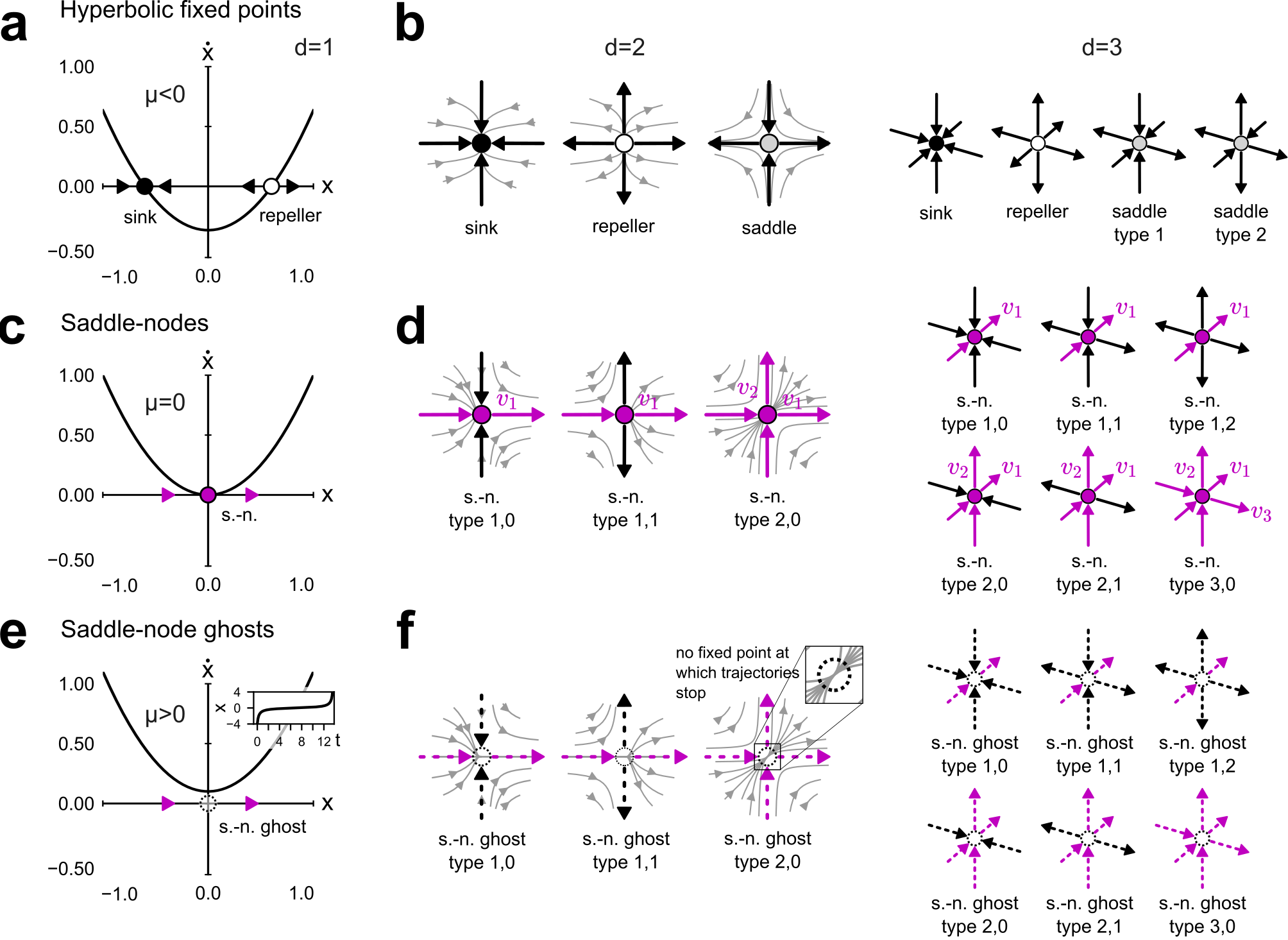}
	\caption{Organization of phase space flow by fixed points and saddle-node ghosts. \textbf{a, b}, hyperbolic fixed points in dimensions 1-3. \textbf{a}, sink and repeller for flows on the line (\autoref{eq:1}) for which stability of fixed points can be visually assessed by considering $\dot x$ versus $x$. \textbf{c,d}, non-hyperbolic saddle-nodes in dimensions 1-3. \textbf{c}, saddle-node at the bifurcation point for flows on the line. Since $\dot x >0$ for all $x\neq 0$, the saddle-node at $x=0$ is semi-stable: the negative half-line is attracted to the fixed point, the positive half-line moves away from it. \textbf{d}, saddle-nodes of type $j,k$ in dimensions 2 and 3. Black arrows indicate eigenvectors of the stable/unstable subspaces, magenta arrows indicate eigenvectors of the center subspace. \textbf{e,f}, ghosts for various types of saddle-nodes in dimensions 1-3. \textbf{e}, ghost for flows on the line. While $\dot x > 0$ for all $x$, the general direction of the flow is identical as for the saddle-node and low values for $\dot x$ at $x\approx0$ (bottleneck) lead to a long transient (inset) which together constitutes finite-time attraction. \textbf{e}, ghosts of type $j,k$ saddle-nodes in dimensions 2 and 3. Dotted arrows indicate flow directions at the ghost. Note how despite the absence of fixed points the phase portrait maintains high similarity to the parental saddle-node.}
	\label{fig:fig1}
\end{figure}
\newpage
\begin{definition}[\textit{Saddle-node equilibrium and bifurcation point of type $j,k$}]\label{def:SN}
Let $\dot{x} = f(x,\rho_{\textnormal{crit}})$, $x\in \mathbb{R}^n$, $f\in \mathcal{C}^\omega$, be an autonomous dynamical system with parameters $\rho_{\textnormal{crit}} = (\rho_1,\dots,\rho_m)^T\in \mathbb{R}^m$. We say $x^*$ fulfilling $f(x^*)=0$ is a saddle-node equilibrium of type $j,k$ and $(x^*,\rho_{\textnormal{crit}})$ a saddle-node bifurcation point of type $j,k$ if the following holds true:
\begin{enumerate}[label={(\textbf{SN\arabic*})}, ref={(\textbf{G\arabic*})}]
    \item $D_xf(x^*,\rho_{\textnormal{crit}})$ has a semi-simple eigenvalue 0 of algebraic and geometric multiplicity $j$, with corresponding linearly independent eigenvectors  $v_1,\dots,v_j$ and $w_1,\dots,w_j$ to the right and left, respectively, that vary continuously with $(\rho_1,\dots,\rho_m)^T$.
    \item For each $i$ with $1\leq i \leq j$ there exists $l$ with $1\leq l \leq m$ with $w_i((\partial f/\partial\rho_l)(x^*,\rho_{\textnormal{crit}})) \neq 0$.
    \item for every non-zero $u\in \ker(D_xf(x^*,\rho_{\textnormal{crit}}))$, there exists $i$ with $1\leq i \leq j$ such that \mbox{$w_i(D_x^2f(x^*,\rho_{\textnormal{crit}})(u,u))\neq 0$}.
    \item $D_xf(x^*,\rho_{\textnormal{crit}})$ has $k$ eigenvalues $\lambda^+_1,\dots,\lambda^+_k$ with positive real parts, and $n-j-k$ eigenvalues $\lambda^-_1,\dots,\lambda^-_{n-j-k}$ with negative real parts.
\end{enumerate}
\end{definition}

For $j=1$, Definition \ref{def:SN} reduces to the definitions given in \cite{Amaral2010,Amaral2012}, but for $j>1$ captures the other saddle-nodes shown in \autoref{fig:fig1}d: (SN1) ensures a center subspace space spanned by $j$ independent eigendirections, (SN2) means that for each center direction, there is a parameter that upon perturbation leads to an unfolding of the saddle-node, i.e. saddle-nodes of type $j,k$ are structurally unstable. Note that since this can include different parameters for different center directions, the saddle-node of type $j,k$ has a codimension of $\leq j$. (SN3) guarantees quadratic non-degeneracy for each non-zero center direction. 
(SN4) simply describes the dimension of the saddle-nodes' stable and unstable subspaces.

Having defined saddle-nodes of type $j,k$, we next will consider how their \emph{ghost} influences the non-asymptotic dynamics of a system after a SN-bifurcation occurs. To do so, we return again to \autoref{eq:1}, this time for a small $\mu>0$ (\autoref{fig:fig1}e). Since $\dot x=\mu+x^2>0$ for $\mu>0$, \autoref{eq:1} has no fixed points anymore, i.e., no invariant objects can govern the system's dynamics. However, while $\dot{x}>0$, $\mu$ is small and so for small $|x|$, $\dot x$ becomes small, too, and any trajectory starting on the negative half-line spends a long time in the vicinity of 0 before eventually escaping towards infinity (cf. inset in \autoref{fig:fig1}e). Notice how the phase portrait of the ghost, except for the vanished fixed point, does not differ much from that of its saddle-node at $\mu$, in particular the direction of the flow remains the same. Indeed, this simple example already shows many features that are associated with ghosts: slow dynamics, non-invariance and a finite-time sense of attraction. However, as saddle-nodes can become more complex in higher dimensions, so do their ghosts. Consider the following system:
\begin{align}
\dot{x}_i = \label{eq:2}
\begin{cases}
\mu_i + x_i^2, & 1 \le i \le j, \\[6pt]
x_i, & j+1 \le i \le j+k, \\[6pt]
-x_i, & j+k+1 \le i \le n
\end{cases}
\end{align}

where integers $n,j,k$ describe the system's dimension, the number of dimensions that follow normal form dynamics, and the number of attracting and repelling directions, respectively, and $\mu_i, 1\leq i \leq j$ are the parameters. By design, \autoref{eq:2} describes saddle-nodes of type $j,k$ in $\mathbb{R}^n$ for any $n$ at $\mu_1=\dots=\mu_j=0$. 
For $\mu_1=\dots=\mu_j>0$, we see ghosts for each saddle-node of type $j,k$ emerging, each of whose phase portrait closely resembles the corresponding saddle-node (\autoref{fig:fig1}f). Ghosts of saddle-nodes of type $1,0$ are the `standard' ghosts commonly found in the literature: funneling trajectories from a larger area of the phase space, slowing them down and ejecting them into a single direction. In contrast, ghosts of type $j,k$ saddle-nodes with $k>0$ repell most trajectories, defying the notion of attraction. Ghosts of type $j,0$ saddle-nodes, $j>1$, inherit the phase-portrait organized by the crossing center eigendirections of the saddle-nodes, collecting and bundling trajectories from one region of the phase space, slowing them down locally, but then ejecting them into more than one directions (cf. s.-n. ghost type $2,0$ in \autoref{fig:fig1}f), in that sense making ghosts of $j,0$ saddle-nodes, $j>1$, genuinely higher-dimensional. While this phenomenological description shows that our standard notions do not apply to all ghosts, we still have not defined formally what a ghost is. Perhaps owing due to their non-invariant and transient nature, there is, in fact, no single and generally accepted formal definition of ghosts in dynamical systems yet. The recent definition by Zheng \textit{et al}. \cite{Zheng2025}, for example, considers ghosts as the non-zero minima of user-defined cost functions to capture the slow dynamics and enable their continuation across parameters, whereas our previous definition \cite{Koch_2024} accounts for their non-invariance and the existence of a basin of attraction to define composite structures made up from multiple ghosts in phase space. However, neither definition captures all of the above described features and is sufficient for the purpose of this study. In particular, local minima of some of the cost functions described in \cite{Zheng2025}, as we will see later, also occur in systems with long transients that are not due to ghosts and neither \cite{Zheng2025} nor \cite{Koch_2024} account for the dimensionality of a ghost. We thus propose the following new definition and terminology:

\begin{definition}[\textit{Ghost of type $j,k$ saddle-node}]\label{def:ghost}
Let $\dot{x} = f(x,\rho)$, $x\in \mathbb{R}^n$, $f\in \mathcal{C}^\omega$, be an autonomous dynamical system with parameters $\rho = (\rho_1,\dots,\rho_m)^T\in \mathbb{R}^m$ and $\mathcal{A}\in\mathbb{R}^n$ a closed bounded set. We say the system has a ghost of a type $j,k$ saddle-node at $x_g\in \mathcal{A}$ if the following holds:
\begin{enumerate}[label={(\textbf{G\arabic*})}, ref={(\textbf{G\arabic*})}]
    \item $x_g$ is the non-zero minimum of $Q:\mathbb{R}^n\to \mathbb{R}$, $Q(x)=\frac{1}{2}||f(x,\rho)||^2$ restricted to $\mathcal{A}$ (\textit{slowness})
    \item $\mathcal{A}$ contains no semi-trajectories in forward time (\textit{non-invariance})
    \item There is a saddle-node $x_{\textnormal{sn}}$
    of type $j,k$ in parameter space at $\rho_{\textnormal{crit}}$ such that $\lim_{\rho\to\rho_{\textnormal{crit}}}x_g=x_{\textnormal{sn}}$ and \mbox{$||\rho-\rho_{\textnormal{crit}}||<||\rho-\hat\rho_{\textnormal{crit}}||$} for any other saddle-node $\hat{x}_{\textnormal{sn}}$ of type $\hat{j},\hat{k}$ at $\hat{\rho}_{\textnormal{crit}}\neq \rho_{\textnormal{crit}}$.  (\textit{parental saddle-node})
\end{enumerate}
We call $j$ the \emph{dimension} of the ghost and say ghost $x_g$ is \emph{attracting} if its limiting saddle-node is of type $j,0$ and \emph{non-attracting} otherwise.
\end{definition}

While we have not explicitly stated how the set $\mathcal{A}$ is to be defined, we find in practice a small neighborhood of a potential ghost to be sufficient. (G1) formally captures the slow dynamics that result from the flat quasi-potential landscape due to the coalescence of equilibria at the saddle-node \cite{Stanoev2018,Koch_2024,Koch2024_rev}. The auxiliary function $Q$ is equivalent to the cost-function $J$ used in \cite{Zheng2025}. (G2) simply states the non-invariance of ghosts, i.e., any trajectory visiting the neighborhood of a ghost must eventually leave this neighborhood again. While (G1-G2) are necessary conditions for ghosts, they are likely met by most long transients. (G3) is the core of the definition and states that each ghost has a closest saddle-node of type $j,k$ in parameter space from which it inherits the organization of the phase portrait and which defines the ghost's dimension. 
It also excludes that a ghost can be a ghost of multiple saddle-nodes and will be important in the context of different ways of unfolding a type $j,k$ saddle-node. Although the notion of attraction is only implicit (as it is not necessary for the purpose of this study), definition \ref{def:ghost} accounts for all features observed for the ghosts highlighted above.

\subsection{The GhostID algorithm}\label{subsec:algorithm}

Having formally defined what we mean by a ghost brings us a step closer to an algorithm that can identify and characterize ghosts. While (G1) and (G2) are straightforward to test, searching the parameter space for parental saddle-nodes is not a feasible strategy. We thus need move on from (G3) to an algorithmically identifiable criterion. For this, we will rely on our previous observation \cite{Koch_2024} that all ghosts we examined so far exhibit a linear spatial distribution of instantaneous eigenvalues ranging from negative to positive along the former center direction of the saddle-node. Indeed, we find the same to be true for higher dimensional ghosts (e.g., for system \ref{eq:2} for $n=j=2,k=0$, cf. Supplementary \autoref{sfig:2dghost}). This leads us to formulate the following conjecture:

\begin{conjecture}[Eigenvalue spectrum of type $j,k$ saddle-node ghosts]\label{eigvalConjecture}
Let $\dot{x} = f(x,\rho)$, $x\in \mathbb{R}^n$, $f\in \mathcal{C}^\omega$, be an autonomous dynamical system with parameters $\rho = (\rho_1,\dots,\rho_m)^T\in \mathbb{R}^m$ for which $x_g\in \mathcal{A}$ is a ghost of a type $j,k$ saddle-node $x_{sn}$ at $\rho_{\textnormal{crit}}$. For each $i$ with $1 \leq i \leq j$, there exists a unit vector $u_i\in \mathbb{R}^n$ such that in the neighborhood of $x_g$, the eigenvalues $\lambda_i$ of $D_xf(x,\rho)$ along the line $\{x\in \mathcal{A}\ |\ x_g+tu_i,\ t\in\mathbb{R}\}$ vary continuously and change sign from negative to positive with $\lambda_i(t)\approx ct$, $c\in\mathbb{R}$, to leading order.
\end{conjecture}

Conjecture \ref{eigvalConjecture} states that the instantaneous eigenvalue distributions along the former center directions of the saddle-node still persist at the ghost, which is what we will use as algorithmic criterion to evaluate if a trajectory passes nearby a ghost. Although we were not able to provide a proof here, we strongly suspect conjecture \autoref{eigvalConjecture} follows from definition \ref{def:SN} and \ref{def:ghost} and find empirically that such a distribution exists in all systems exhibiting a ghost state resulting from a SN-bifurcation. While such a spatial distribution of instantaneous eigenvalues might potentially result from other mechanisms, too, and hence does not strictly guarantee the presence of a ghost, we empirically find it to be a very strong indicator based on numerous tested systems with different mechanisms underlying long transients. 

\begin{wrapfigure}{r}{0.45\textwidth}
  \begin{center}
    \includegraphics[width=0.44\textwidth]{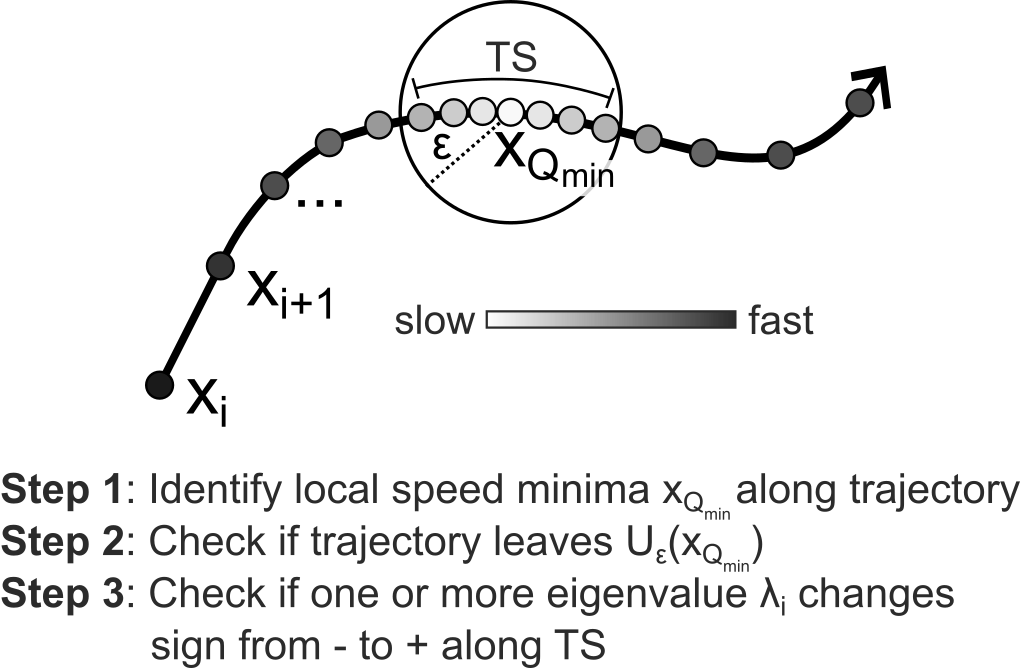}
  \end{center}
  \caption{Illustration of the main steps of how ghosts are identified via the GhostID algorithm.} \vspace{0.5cm}
  \label{fig:fig2}
\end{wrapfigure}

Following the outlined framework, we propose a simple algorithm, which we term GhostID, that identifies ghosts along a trajectory by evaluating criteria (G1), (G2) and the changes in instantaneous eigenvalues from conjecture \autoref{eigvalConjecture}. The basic steps of this algorithm, illustrated in \autoref{fig:fig2}, are as follows: The first step of the algorithm is to find the locally slowest points along a trajectory, corresponding to evaluating a trajectory proxy for criterion (G1). In the second step, the algorithm evaluates if the trajectory leaves the $\varepsilon$-neighborhood of a locally slowest point, testing for non-invariance (G2). The third and final step consists of testing whether one or more of the instantaneous eigenvalues change from negative to positive along the trajectory segment within the $\varepsilon$-neighborhood from the second step. Thus, starting from a slow point, if both non-invariance and the eigenvalue criterion are satisfied, we consider the long transient to be caused by a ghost, with the number of eigenvalues changing sign being considered by us as the ghost's dimension. A full technical description of the algorithm, its hyper-parameters and some rationales for choosing their values them can be found in Supplementary section \ref{subsec:ghostIDparas} and Supplementary \autoref{sfig:algorithm}a-b.

Besides the GhostID algorithm, we developed three further algorithms that are build on GhostID and represent core functionalities of the PyGhostID package (Supplementary \autoref{sfig:algorithm}c): {\tt ghostID$\_$phaseSpaceSample}, an integrated simulation and evaluation version of the algorithm that uses Latin hypercube sampling and parallel processing to sample many trajectories within a phase space region for ghosts, {\tt ghost$\_$connections}, which allows to reconstruct ghost channels, ghost cycles \cite{Koch_2024} and more complex composite structures of ghosts from one or more sequences of ghosts identified by GhostID, and {\tt track$\_$ghost$\_$branch}, which, inspired by \cite{Zheng2025}, allows to track an identified ghost state as system parameters are varied. While a more detailed description of these functions can be found in Supplementary section \ref{subsec:PyGhostIDpackage} and the PyGhostID documentation, we provide examples of their application in \autoref{subsec:applications}.

\subsection{Numerical validation of the GhostID algorithm}\label{subsec:validation}

Before applying GhostID to obtain new insights, we first perform a series of numerical validations to ensure the algorithm correctly identifies ghosts in models from various disciplines and does not identify long transients as ghost that result from other mechanisms such as slow manifolds from slow-fast systems or saddle crawl-bys.

To begin our validation, let us consider the following model of ecological interactions between corals and macroalgae by Bieg \textit{et al}. \cite{Bieg_2024}:
\begin{align}
        \dot{C} &= C \left( r(1 - C N_t) B - m - a M \left( \frac{N_t}{N_0 + N_t} \right) \right) \label{eq:bieg}\\
    \dot{M} &= M \left( a C \left( \frac{N_t}{N_0 + N_t} \right) - \frac{g}{M + B} + \gamma B \right) \nonumber\\
    B &= 1 - C - M, \nonumber
\end{align}
where $C$, $M$ and $B$ represent corals, macroalgae and benthic r-strategists (groups of algae that colonize disturbed reefs on a faster timescale than macroalgae), and parameters $a,b, r$ and $\gamma$ represent (over-)growth rates, $g$ the grazing rate by herbivores, and $m$ is the mortality rate of corals.
This model has been proposed to be organized close to a saddle-node bifurcation where it exhibits a ghost attractor to explain how coral reef recovery may be hampered by repeated climate or human driven perturbations that prevent trajectories escaping a long transient of the algae-dominated state \cite{Bieg_2024}. Applying GhostID to a trajectory from this regime correctly identifies a ghost between the system's nullclines (\autoref{fig:fig4}a, left). Considering the $pQ$ time series reveals that the algorithm interestingly identified two $Q$-minima along the trajectory (\autoref{fig:fig4}a, middle). Once the trajectory reaches the system's stable fixed point, $pQ(t)$ appears to fluctuate irregularly, leading to additional $Q$-minima. However, these irregular movements are mere numerical noise which can be reduced by lowering the numerical tolerances of the integration procedure (Supplementary \autoref{sfig:numNoise}) and have been ignored by GhostID by requiring identified peaks to have a minimum width of $50$dt. The eigenvalues along the trajectory segment around the first identified $Q$-minimum are consistent with the system visiting a slow region around the saddle on the y-axis and are thus not identified as a ghost, whereas on the segment of the second minimum eigenvalue $\lambda_2$ changes sign from negative to positive, leading to the algorithm identifying a ghost (\autoref{fig:fig4}a, right). Additional examples of GhostID correctly identifying known ghosts in models from cellular biochemistry, gene regulatory networks, and condensed matter physics can be found in Supplementary section \ref{subsection:numVal_ghosts} and Supplementary \autoref{sfig:numval_ghosts}. Lastly, while we are not aware of higher-dimensional or non-attracting ghosts being described in the literature, GhostID correctly identifies these ghosts and their corresponding dimensions from our example system \ref{eq:2} (Supplementary \autoref{sfig:numval_ghosts_new}).

Having confirmed that GhostID correctly identifies previously characterized ghost states in models of various real-world phenomena, we next test if GhostID correctly ignores long transients resulting from other phenomena. First, we test GhostID on trajectories with long transients resulting from passing near saddles. The first system we consider for this task is a simple ecological model from Hastings \textit{et al.} \cite{Hastings_2018}:

\begin{equation}
    \dot{N} = \alpha N \left(1 - \frac{N}{K}\right) - \gamma \frac{N P}{N + h} \label{eq:Hastings}, \quad
    \dot{P} = \epsilon \left( v \gamma \frac{N P}{N + h} - m P \right), 
\end{equation}

where $N$ denotes the prey species, $P$ the predators, and $\alpha$ is the growth rate of the prey population, $K$ its carrying capacity, $\gamma$ the predation rate, $h$ the half-saturation constant for predation, $v$ is a scaling factor of how predation influences predator reproduction, $m$ is the mortality rate of predators and $\epsilon$ the timescale separation between predator and prey dynamics. For the parameter regime depicted in \autoref{fig:fig4}b (left), the system exhibits oscillatory dynamics in which the orbit is forced to pass close to two different saddles (saddle crawl-by), causing the trajectory to slow down in their vicinity. However, the eigenvalues along the trajectory segments of the corresponding $Q$-minima do not cross $0$, hence the saddle crawl-bys are not identified as ghosts by our algorithm (\autoref{fig:fig4}b, middle and right). Another example involving a heteroclinic cycle can be found in Supplementary section \ref{subsection:numVal_saddles} and Supplementary \autoref{sfig:numval_HC}. Since saddles are hyperbolic fixed points by definition, it is easy to prove that within a sufficiently small neighborhood of the fixed point, the eigenvalues along a trajectory in this neighborhood can never change sign (Supplementary section \ref{subsec:theorem_vicinityHypFP}) and thus will not be identified as ghosts by our algorithm.

The next important mechanism to test our algorithm against are long transients in slow-fast systems generated by an explicit separation of timescales via a small parameter $0<\epsilon\ll1$. The first system we consider is the FitzHugh-Nagumo model, a simplified version of the original Hodgkin-Huxley neuronal model. Here, we use the dimensionless version of the system by Cebrián-Lacasa \textit{et al}. \cite{CebrinLacasa2024}, given by:

\begin{equation}
    \dot{u} = u - u^3 - v , \quad \dot{v} = \epsilon (u - b v + a), \label{eq:FHN}
\end{equation}
where $u$ represents membrane voltage of the neuron, $v$ slow recovery from depolarization through the opening of K$^+$- and inactivation of Na$^+$-channels, $a,b$ are dimensionless parameters, and $\epsilon$ the timescale separation. Using GhostID on a trajectory from \autoref{eq:FHN} in a relaxation oscillation regime with pronounced timescale separation shows again that our algorithm does not identify the slow dynamics from slow-fast systems as ghosts (\autoref{fig:fig4}c). Further examples of GhostID correctly non-identifying long transients stemming from slow-fast dynamics rather than ghosts can be found in Supplementary section \ref{subsection:numVal_slowFast} and Supplementary \autoref{sfig:numval_SF}.

Having tested our algorithm on various models from different disciplines exhibiting different types of long transients, we conclude that GhostID correctly and specifically identifies ghost dynamics.

\begin{figure}
	\centering
	\includegraphics[width=\textwidth]{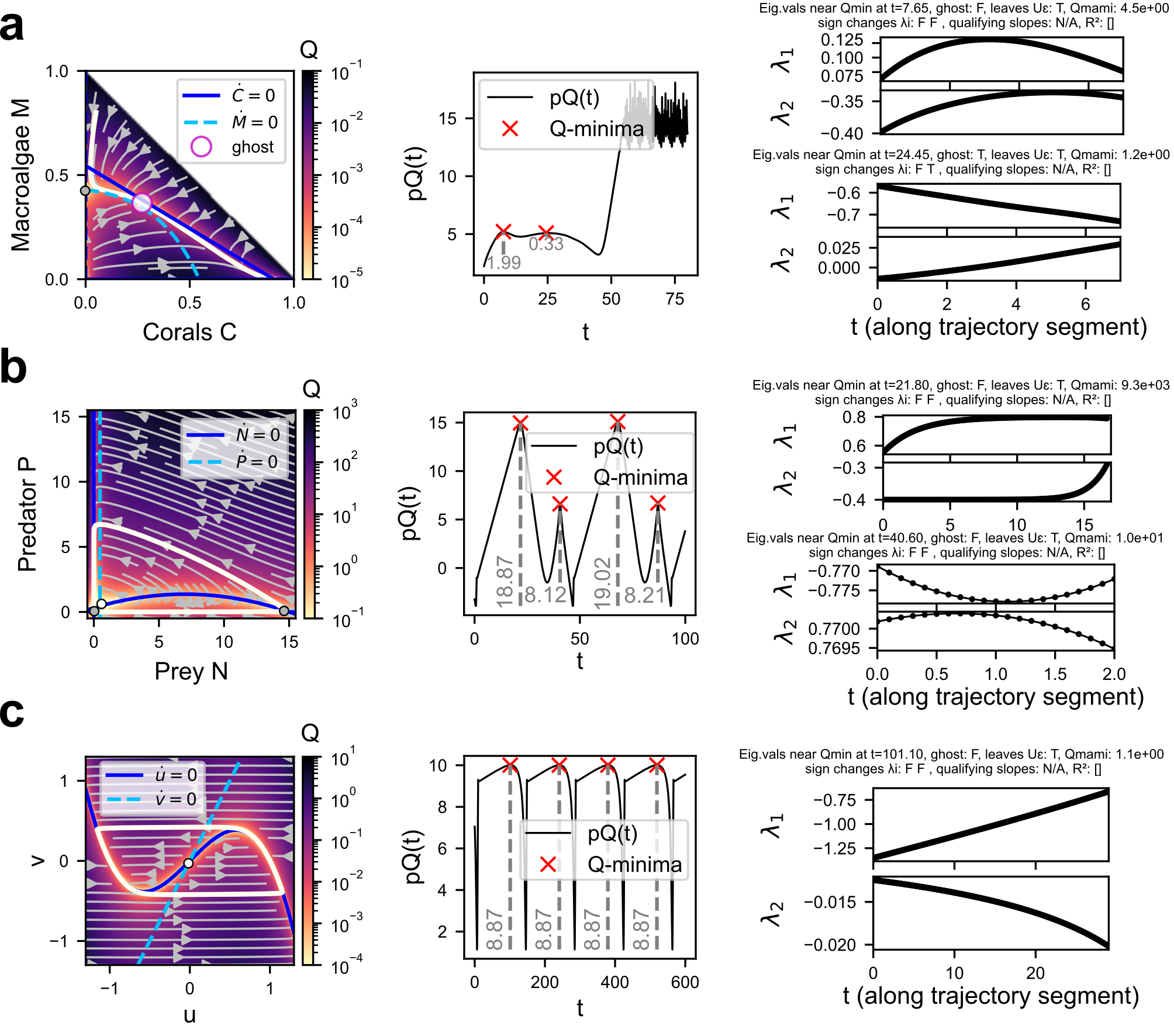}
	\caption{Numerical validation of GhostID. \textbf{a}, ghost identified by GhostID in ecological model from \cite{Bieg_2024}. Parameter values: $a = 2.0$, $\gamma = 0.7$, $m = 0.15$, $g = 0.4$, $N_t = 0.53$, $N_0 = 0.5$, $r = 1.8$, $c = 0.25$.
    \textbf{b}, long transients due to saddle crawl-bys in model from \cite{Hastings_2018} are not identified by GhostID. Parameter values: $\gamma=2.5$, $h=1$, $v=0.5$, $m=0.4$, $\alpha=0.8$, $K=15$, $\epsilon=1$. \textbf{c}, long transients due to slow-fast dynamics in FitzHugh-Nagumo model are not identified by GhostID. Parameter values: $a=0$, $b=0.5$, $\epsilon=0.01$. \textbf{a-c}, Left: phase space. Grey dots indicate saddles, white dots indicate unstable focus. Middle: $pQ(t)$ plots along with $Q$-minima and peak prominences. Right: eigenvalues along trajectory segments and evaluation criteria.}
	\label{fig:fig4}
\end{figure}

\subsection{Applications}\label{subsec:applications}
Having defined new types of ghosts and algorithms to identify and characterize them, we seek to find examples of these phase-space objects. For this purpose, we next turn to paradigmatic models from several disciplines in which long transients and ghost dynamics have received still relatively limited attention.

\subsubsection{Two-dimensional ghosts and ghost-bifurcations in coupled theta neurons}

To begin our investigation into complex and potentially novel ghost dynamics, we decided to focus on simple models of coupled neurons whose dynamics in isolation exhibit a SNIC bifurcation as key ingredient for the emergence of ghosts. While long transients and metastability in computational neuroscience have received some attention over time (cf. \cite{Hancock2024,Rossi_2025} for reviews), the role of ghosts in neuronal dynamics remains relatively underexplored. A good example for the importance of ghosts was recently provided by Choi \textit{et al.} which proposed fly olfactory neurons to be organized close to a SNIC bifurcation, showing that the resulting ghost dynamics provide optimal information processing for sensory stimuli \cite{Choi_2024}. We thus consider the theta neuron \cite{Ermentrout1986} as a simple phase oscillator model that captures the essential dynamic features described in \cite{Choi_2024} to explore how pulse coupling of two identical theta neurons influences the ghost dynamics. The equations for this systems are given by \cite{Augustsson2024}: 
\begin{align}
    \dot{\theta}_1 &= 1 - \cos(\theta_1) + \left(1 + \cos(\theta_1)\right) \left(\eta_1 + \kappa \left(1 - \cos(\theta_2)\right)\right) \label{eq:augustson}\\
    \dot{\theta}_2 &= 1 - \cos(\theta_2) + \left(1 + \cos(\theta_2)\right) \left(\eta_2 + \kappa \left(1 - \cos(\theta_1)\right)\right) \nonumber,
\end{align}
where $\theta_i$ is the phase, $\eta_i$ the neuron's excitability threshold, and $\kappa$ the coupling strength determining the synaptic input current from neighboring neurons. In isolation, each theta neuron can either be excitable ($\eta<0$) or generate periodic spikes through a SNIC-bifurcation when $\eta>0$. For the coupled system we assume a coupling strength of $\kappa = 0.1$. We first consider the case of two identical neurons, i.e., $\eta_1=\eta_2=\eta$. For $\eta<0$ we find the same phase space topology with four fixed points as for \autoref{eq:2} (\autoref{fig:fig6}a, left). Increasing $\eta$ leads to a saddle-node bifurcation at $\eta=0$, where all four fixed points coalesce at the origin (\autoref{fig:fig6}b).  Examining this equilibrium shows that it is consistent with definition \ref{def:SN} and can be classified as a type $2,0$ saddle-node (Supplementary section \ref{subsec:theorem_jKSNinThetaNeurons}).
Using GhostID on a trajectory that passes through the slow region of the phase space for $\eta>0$ identifies a ghost of dimension 2 in the center of the phase space in which the flow is first funneled, then spread out again (\autoref{fig:fig6}a, right). Note that since the phase space of \autoref{eq:augustson} is the torus, no stable periodic orbit exists in this topology. Using PyGhostID's {\tt track$\_$ghost$\_$branch} function, we can track the identified ghost as we vary the excitability threshold $\eta$ and include the ghost in the bifurcation diagram (\autoref{fig:fig6}b). Moreover, we can extract information about each ghost along the branch and, as an example, have visually represented the trapping time at the ghosts by color, allowing us to easily identify in which parameter regime to expect significant long transients.

\begin{figure}
	\centering
	\includegraphics[width=\textwidth]{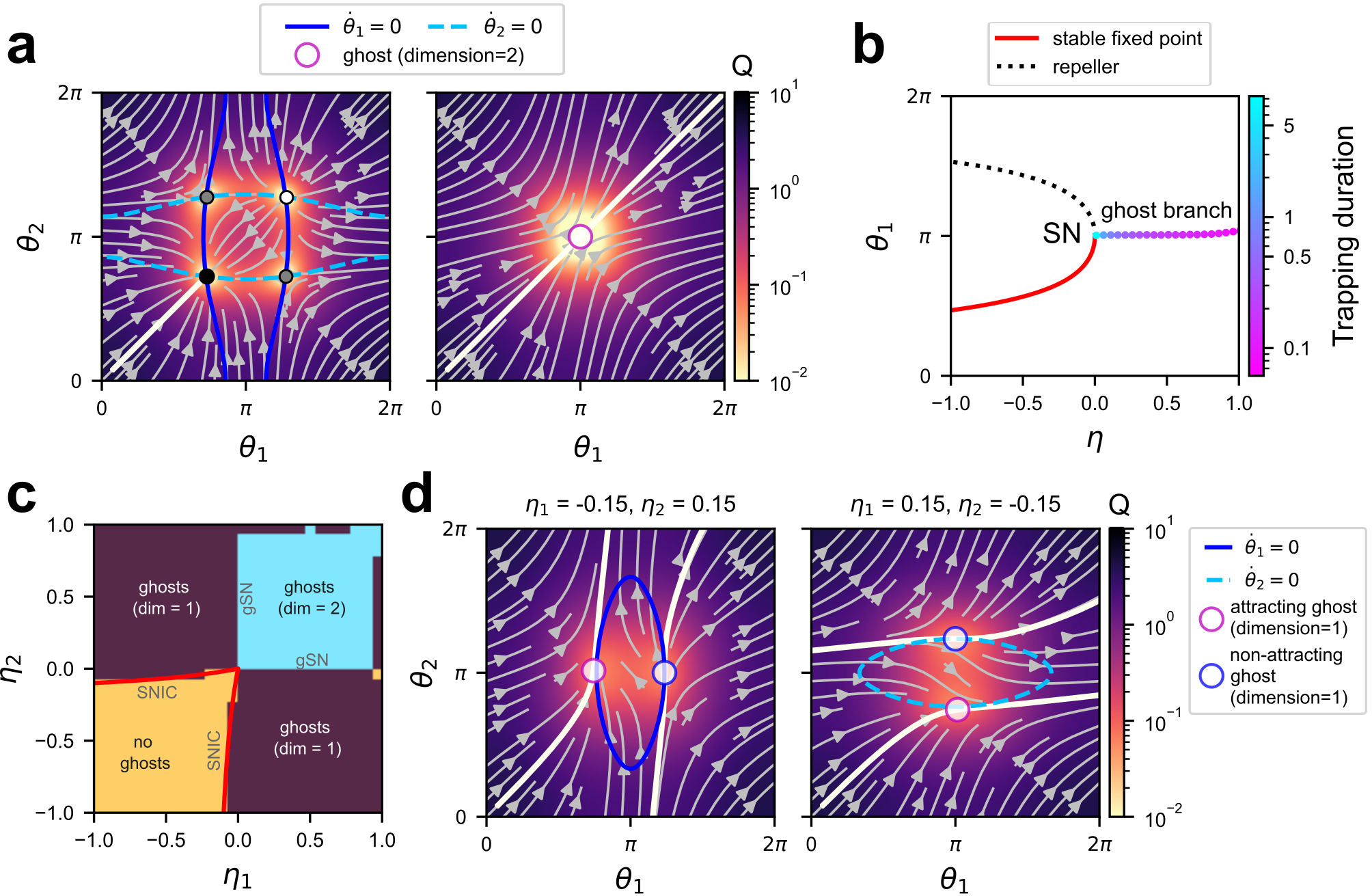}
	\caption{Higher-dimensional ghosts in coupled theta neuron models. \textbf{a}, left: phase space organization of identical theta neurons before bifurcation point. Black dot indicates sink, white dot indicates repeller, grey dots indicate saddles. Right: trajectory and two-dimensional ghost identified by GhostID after bifurcation. \textbf{a,d} White lines indicate trajectories, grey arrows flow. Note that for better visualization, $\theta_1$ and $\theta_2$ have been phase-shifted by $\pi$ in these plots, such that the origin is shown at $(\pi,\pi)$. \textbf{b}, tracking of the ghost identified in \textbf{a} in a bifurcation diagram using the {\tt track$\_$ghost$\_$branch} function. Each ghost along the branch is color-coded by its trapping duration. \textbf{c}, two-parameter bifurcation diagram for non-identical theta neruons with distinct excitability thresholds. Boundary between areas with one- and two-dimensional ghosts indicate \textit{ghost SN-bifurcation} (gSN). Red line indicates SNIC bifurcation. \textbf{d}, phase portraits for parameters yielding one-dimensional ghosts. 
    }
	\label{fig:fig6}
\end{figure}

Next, we turn to the case of non-identical neurons, i.e., allowing for $\eta_1\neq\eta_2$. While {\tt track$\_$ghost$\_$branch} currently does not feature tracking of ghosts versus two parameters, we can apply the {\tt ghostID$\_$phaseSpaceSample} method to characterize the ghost dynamics across the $\eta_1$-$\eta_2$-plane from a Latin hypercube sample within a specified region in phase space. Using this approach, we find two transitions from a parameter region in which only fixed points exist (i.e. without any ghosts) towards regions in which only one-dimensional ghosts exist after crossing a SNIC bifurcation (\autoref{fig:fig6}c). 
Both SNIC branches merge at $\eta_1=\eta_2=0$, corresponding to two saddle-nodes of type $1,0$ coalescing at $\theta_1=\theta_2=0$ to form a type $2,0$ saddle-node. After passing one of the SNIC bifurcations, the system exhibits a stable and an unstable limit cycle due to one of the neurons still being in an excitable, not spiking regime (\autoref{fig:fig6}d). Interestingly, moving from the 1-dimensional ghost region into the region defined by both $\eta_1,\eta_2>0$ reveals another transition into a region where two-dimensional ghosts exist. 
This transition seems to be best described as a \textit{ghost SN-bifurcation}, in which the attracting one-dimensional ghost merges with the non-attracting one-dimensional ghost shown in (\autoref{fig:fig6}d). Indeed, we observe the same bifurcation between the ghost of a type $1,0$ saddle-node and the ghost of a type $1,1$ saddle-node for system \ref{eq:2} (Supplementary \autoref{sfig:gSN_sys2}). Together these results highlight the capability of GhostID to detect ghosts, track them across parameter space and identify transitions in their properties.

\begin{figure}
	\centering
	\includegraphics[width=\textwidth]{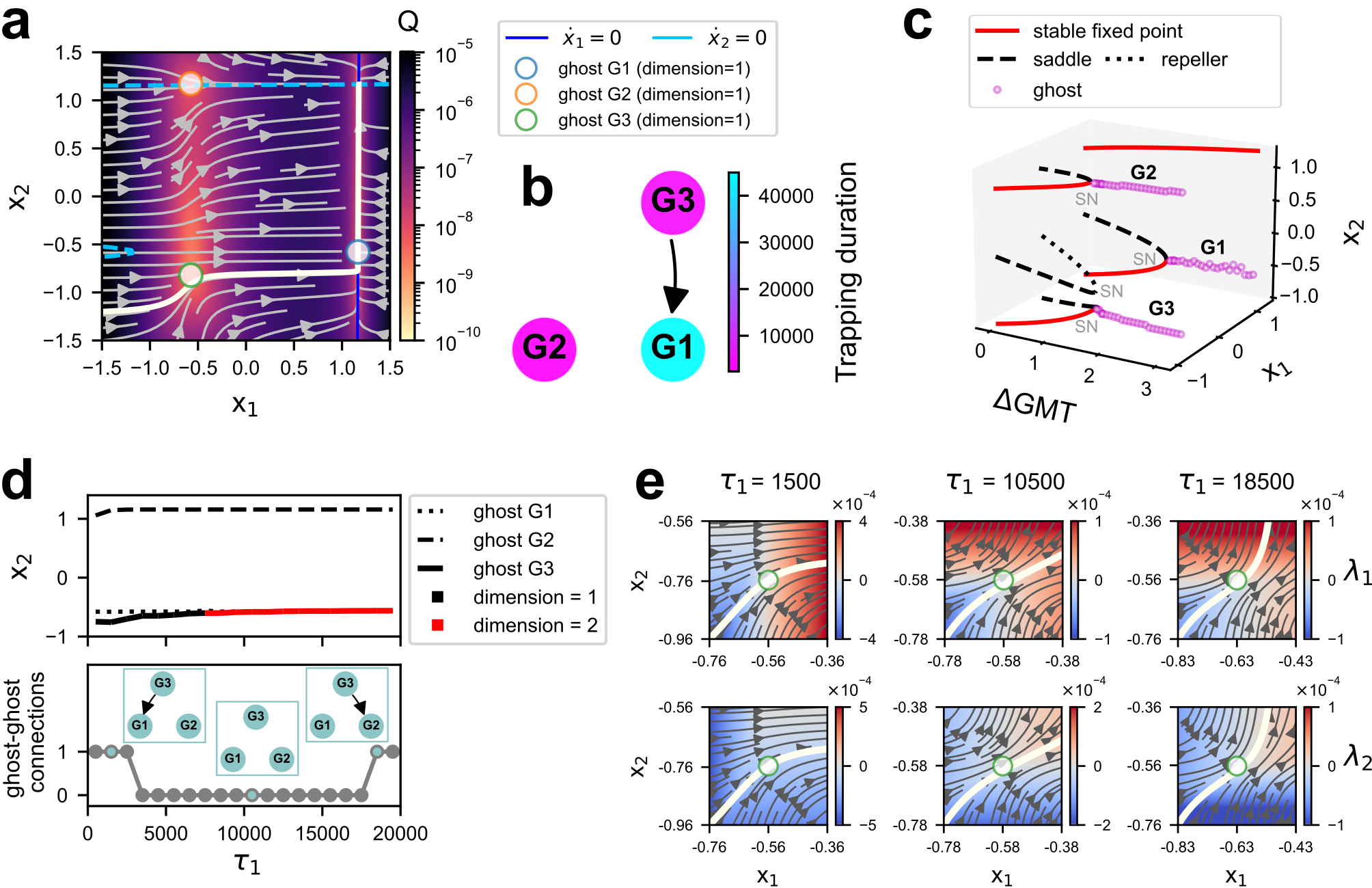}
	\caption{Ghosts and ghost structures in coupled climate tipping elements. \textbf{a}, phase portrait and ghosts identified from a phase space sample. \textbf{b}, ghost connections identified by PyGhostID's {\tt ghost$\_$connections} function showing a ghost channel between G3 and G1. \textbf{c}, tracking of ghosts from \textbf{a} versus $\Delta GMT$. \textbf{d}, influence of timescale parameter $\tau_1$ on ghost dynamics, showing a transition in G3's dimension (top) and a restructuring of the ghost channel from G3$\to$G1 to G3$\to$G2 (bottom). \textbf{e}, local redirection of the flow and changes in the spatial distribution of instantaneous eigenvalues in the vicinity of G3 for different values of $\tau_1$.}
	\label{fig:fig7}
\end{figure}

\subsubsection{Timescale induced transitions in ghosts of coupled climate tipping elements}

As another example, we consider coupled climate systems. While tipping of climate systems due to passing through saddle-node bifurcations is a key topic in climate research, the number of studies explicitly considering the role of ghosts is surprisingly sparse. Only recently, for example, has the potential role of ghost dynamics after tipping of the Atlantic meridian overturning circulation been taken into account \cite{Mehling2024,borner2025boundary}.
For studying ghost dynamics in tipping elements, we decided here to use a model of cascading climate tipping elements by Wunderling \textit{et al.} \cite{Wunderling2021} which is conceptually simple, yet easy to parametrize with the relevant real-world parameters and has informed a series of policy-relevant studies on climate change \cite{Wunderling2022,Moller2024}. The system is given by:

\begin{equation}
    \dot{x}_i = \frac{1}{\tau_i} \left( x_i - x_i^3 + \sqrt{\frac{4}{27}} \frac{\Delta GMT}{T_{\text{limit},i}} + d \sum_{ \forall j: j\neq i}  \frac{s_{ij}}{10}(x_j + 1) \right),\label{eq:Wunderling}
\end{equation}

where $x_i$ denotes the current value of the $i$th tipping element, $\tau_i$ is its tipping timescale, $\Delta GMT$ is the increase in the global mean surface temperature above pre-industrial levels, $T_{\text{limit},i}$ is the critical temperature above which the element tips (passes a SN-bifurcation), $d$ is an overall interaction strength parameter and $s_{i,j}$ represents the link strength between two interacting tipping elements. Since \autoref{eq:Wunderling} explicitly features multiple timescales, it is also an interesting system to explore how slow-fast dynamics and ghost dynamics interact. While a detailed investigation of the ghost structures in this system is the subject of another study [Nandan \textit{et al.}, in preparation], we focus here on two generic tipping elements which mutually reinforce each other. We assume the two tipping elements have similar but not identical critical temperatures and that their tipping timescales are different by a factor of 10. 
Generating several trajectories that sample the phase space using PyGhostID's {\tt ghostID$\_$phaseSpaceSample} function, we identify three ghosts of dimension 1 for $\Delta GMT$ being set just above the highest value of $T_{\text{limit},i}$ (\autoref{fig:fig7}a). Following a trajectory starting at $x_1(0)=-1.5, x_2(0)=-1.2$, we see the trajectory visiting both ghosts G3 and G1 before arriving at a stable fixed point (\autoref{fig:fig7}a, white line), suggesting the system exhibits a ghost channel \cite{Koch_2024} between G3 and G1. To confirm this, we use the {\tt ghost$\_$connections} function on the ensemble of trajectories to generate an adjacency matrix of ghost-ghost sequences. In addition to plotting the corresponding graph, we can encode information about each ghost captured by GhostID visually via the node color, e.g., to represent the trapping duration at each ghost (\autoref{fig:fig7}b). To confirm that {\tt ghostID$\_$phaseSpaceSample} identified all ghosts, we further followed each ghost while varying $\Delta GMT$ values using {\tt track$\_$ghost$\_$branch}, showing that each ghost branch terminates at one of three SN-bifurcation points in which a stable fixed point and a saddle merge (\autoref{fig:fig7}c). 

Next, we explored the influence of tipping timescales in this context. We thus characterize the system's ghost dynamics as we systematically increase $\tau_1$ until $\tau_1>\tau_2$. As {\tt track$\_$ghost$\_$branch} only tracks a single ghost, we used {\tt ghostID$\_$phaseSpaceSample} in combination with {\tt ghost$\_$connections} to identify the dimensions and the connections between the ghosts. While two of the ghosts exhibited no changes, we found a transition in the dimension of G3 from 1 to 2 at $\tau_1 \approx 1.5\ \tau_2$ (\autoref{fig:fig7}d, upper panel). Interestingly, considering the number of ghost connections as quantified by the adjacency matrix reveals another transition: as $\tau_1$ increases, the number of ghost connections drops to 0 at $\tau_1=3500$ before rising to 1 again at $\tau_1\geq18500$ (\autoref{fig:fig7}d, lower panel). Visualizing the corresponding graphs reveals that these transitions are associated with a re-wiring of the ghost connections. While for low values of $\tau_1$, there is a ghost channel from G3 to G1, a new ghost channel from to G3 to G2 is formed for high $\tau_1$ values. While these two types of transitions are clearly different from each other as indicated by different critical thresholds, they are also different from the ghost SN-bifurcation described above, in which an attracting and a non-attracting ghost merge. Zooming in into the phase space region around G3, we can see how increasing $\tau_1$ changes both the distribution of eigenvalues across space and the direction of the flow around G3 (\autoref{fig:fig7}e). At low $\tau_1$-values, we see only eigenvalue $\lambda_1$ around G3 being distributed from negative to positive along the flow, while $\lambda_2$ stays purely negative around G3. As $\tau_1$ increases, the flow compressed in G3 is more spread in the outbound region of G3 where now $\lambda_2$, too, becomes positive, thereby allowing trajectories to sample $\lambda_1$ from negative to positive.
Hence, the change in the dimension of G3 is as result of both the changed eigenvalue distribution across space and of the altered flow, which together determine what eigenvalues a trajectory will see as it passes close to or through G3. The transition between ghost connections, on the other hand, is a result only of the altered flow that determines if a another ghost can be reached starting from G3 or not. 

\subsubsection{Ghost networks in gene regulatory networks at criticality}

As a last example we consider ghost dynamics in GRNs in which each gene exhibits positive auto-regulation and can further be activated and repressed by other genes, modeled in an additive and multiplicative fashion by Hill-equations \cite{Alon2019} with a Hill-coefficient of 2, respectively:

\begin{equation}
    \dot{x_i} = (\frac{\alpha x_i^2}{x_i^2+K^2} + \sum_{\forall j: a_{ji=1}} \frac{\beta x_j^2}{x_j^2+K_a^2})\prod_{\forall j: a_{ji=-1}}\frac{K_i^2}{x_j^2+K_i^2} - x_i\label{eq:GRN},
\end{equation}

where $x_i$ refers to the concentration of gene product $i$, $\alpha$ is the auto-activation strength, $\beta$ the maximum transcription rate of activators, $a_{ij}$ are the elements of the adjacency matrix $A$ that describes the GRN's interactions, and $K,K_a,K_i$ are the half-activation constants of (auto-)activation and repression, respectively. For $K=0.5$ and in the absence of any activators and inhibitors, each gene exhibits a SN-bifurcation at $\alpha=1$ (\autoref{fig:fig8}a). We assume here $\alpha=0.998$, i.e., each gene is organized close to the SN-bifurcation and may thus exhibit a simple ghost, not unlike the situation described in \cite{Farjami_2021}. Directed Erd\H{o}s--R\'enyi networks for which each activatory link was replaced with an inhibitory link with a probability of $0.5$ were used to generate different GRN topologies. 

Using this simple network model, we observed larger ghost structures including ghost channels, ghost cycles and various combinations thereof, with many ghosts having a dimension larger than one. As an example we show a network with $N=12$ nodes (\autoref{fig:fig8}b). Using a large phase space sample of this network generated by {\tt ghostID$\_$phaseSpaceSample}, we identified the ghost structures embedded in the GRN using the {\tt ghost$\_$connections} function and colored each node in the corresponding graph according to the ghost dimensions. As shown in \autoref{fig:fig8}c, we find a strikingly complex ghost network embedded in the phase space of the GRN, featuring many interesting motifs, including cyclical structures, several ghosts converging into one, and ghosts ejecting their flow to two or more other ghosts (e.g. G6). While a systematic characterization of such ghost networks in phase space and their relation to the network structure of the dynamical system is beyond the scope of the current study, a few interesting observations can be drawn from this example already. First, the ghost network has a different topology and is notably larger than the GRN, featuring both more nodes and edges. The topology of the ghost network further suggests multi-stability between a three-ghost cycle, at which the majority of nodes eventually converge by following their links, and two ghost channels from G5 and G9 to G10 after which trajectories converge at an attractor. Timecourses of trajectories starting at three random initial conditions confirm this, showing convergence to the three ghost cycle with its periodic switching between distinct long transients of gene expression after different initial transients for two initial conditions and a two-ghost channel that ends at a stable steady state for the third initial condition (\autoref{fig:fig8}d). Another interesting observation is that the dimensions of the ghosts in the ghost network are relatively small compared to the dimension of the system ($n=12$), with the majority of ghosts having a dimension of 2-4. The dimension of the ghosts furthermore do not necessarily correspond to the number of connections a ghost forms. For example, while G9 has a ghost dimension of four, it only connects to two other ghosts. Sampling initial conditions around this ghost and projecting the high-dimensional trajectories into a three-dimensional principle component space shows that trajectories are being spread into four different directions, consistent with the ghost dimension measured by GhostID, before two pairs of the four trajectory bundles are each being funneled towards G4 and G10, respectively (\autoref{fig:fig8}e).

\begin{figure}
	\centering
	\includegraphics[width=\textwidth]{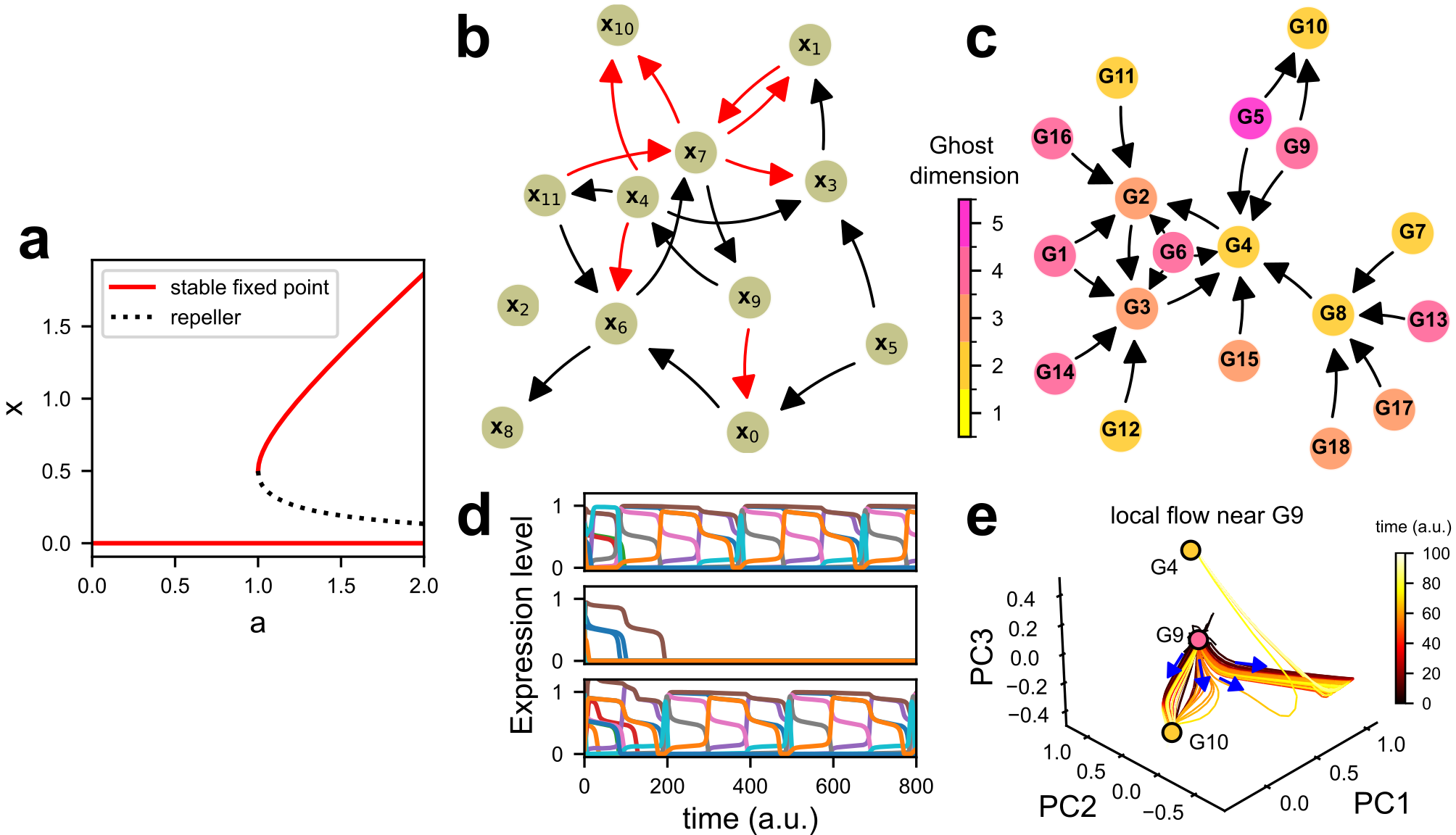}
	\caption{Complex ghost structures in gene regulatory networks. \textbf{a}, bifurcation diagram for single gene dynamics described by equation \ref{eq:GRN}. \textbf{b}, topology of a GRN generated by the Erd\H{o}s--R\'enyi random model. \textbf{c}, toplogy of a complex ghost network embedded in the dynamics of the GRN shown in \textbf{b}. Network was identified from a large phase space sample (300 trajectories) using the {\tt ghostID$\_$phaseSpaceSample} and {\tt ghost$\_$connections} functions from PyGhostID. \textbf{d}, timecourses of the GRN shown in \textbf{b} for three different random initial conditions. \textbf{e}, trajectories shown in principal component space for multiple initial conditions starting close to G9.
    }
	\label{fig:fig8}
\end{figure}

\section{Discussion}\label{sec:discussion}
We have presented here a generalized definition of ghosts in dynamical systems that accounts for the diverse topological structures of the saddle-node fixed points they are the remnant of and developed a series of algorithms to identify and characterize ghosts and their composite structures such as ghost channels and ghost cycles \cite{Koch_2024}. 
Naturally, there are several open questions and limitations that come with our approach. First, our current framework was developed to capture ghosts of fixed points. While further work is necessary to see whether it applies to ghosts of continuous attractors \cite{Fontich2022,Sagodi2025}, the current algorithms cannot capture ghosts of limit cycles as the slow dynamics only apply along the radial direction of their periodic orbit and it is unclear if any meaningful eigenvalue characteristics can be extracted from ghosts of limit cycles. The same argument applies to ghosts from chaotic attractors \cite{Mehling2024,borner2025boundary}. At least the first problem has been solved by the elegant study from Zheng \textit{et al.} \cite{Zheng2025}, who used variational methods to create a `best-fit' topological structure in the vector field. An alternative approach may be to replace the scalar function $Q$ by recurrences \cite{Datseris2022} for finding non-invariant structures in phase space that a trajectory spends a long time in. Secondly, the framework was designed for autonomous deterministic systems. While we suspect that slow parameter drifts will not impede detection of ghosts by our algorithm, faster or more complex parameter changes in non-autonomous systems will likely interfere with the algorithm. Similar considerations apply to stochastic systems. Noise has a significant influence on a trajectories' dwell-time at a ghost and can push it back and forth along the ghost region \cite{Sardanys2020,Koch_2024} which may result in non-monotonic eigenvalue samples that impair ghost identification. Lastly, we assumed throughout that $f$ is analytic. While this assumption holds for a vast number of models from applications, additional assumptions are likely required for our framework to apply to smooth systems (cf. Supplementary section \ref{subsec:smoothnessIssue}), although these assumptions are unlikely to affect the here introduced algorithms in most cases.

Despite these limitations, we showed that our definitions and algorithms offer a powerful new approach to study ghost dynamics in arbitrary dynamical systems. Using several relatively simple example systems, this approach allowed us to discover several interesting phenomena related to ghosts. While their comprehensive characterization is beyond the scope of the current study, we will briefly discuss their implications.

In our first example, two coupled and identical theta neurons \cite{Augustsson2024}, we showed the existence of a two-dimensional ghost following the merging of a sink, a repeller and two saddles (a phase portrait sometimes called \emph{basic tartan} \cite{Baesens2013}) at a type $2,0$ saddle-node of a degenerate SNIC-bifurcation. Since the theta neuron model is the normal form of a SNIC bifurcation, there is strong reason to believe that this result will apply to many coupled systems whose units individually undergo SNIC bifurcations. Characterizing the ghost dynamics in the case of non-identical neurons with different excitability thresholds, we found that one- and two-dimensional ghosts are detectable across a large area of the parameter plane, although pronounced long transients are only found close to bifurcation. This two-parameter analysis further revealed the existence of what we termed a \textit{ghost SN-bifurcation}, a collision of two one-dimensional ghosts - one of which is attracting, one non-attracting in the sense of definition \ref{def:ghost} - leading to the formation of a two-dimensional ghost. To our knowledge, this is the first description of a bifurcation that occurs exclusively between non-invariant objects and produces another non-invariant object. From a biological perspective, the example raises interesting new questions about the role of ghost dynamics in the context of information processing and computations in neuronal networks. For example, could higher dimensional ghosts provide a reliable mechanism for decoding multi-modal time-varying inputs as opposed to single time-varying inputs considered previously \cite{Choi_2024}? Assuming larger ghost networks as found in our GRN example can also emerge from neuronal dynamics, another intriguing perspective is to consider ghosts and their composite structure as a potential transient-based mechanism for structuring phase space trajectories along neuronal manifolds \cite{Perich2025}. As we find the properties of ghosts to be tunable by control parameters, it is tempting to speculate that neuronal networks might be able to adapt the properties of individual ghosts and larger ghost structures to their computational needs. Studying how ghost dynamics structure and support biological computations is thus an important area for further research.

In the second example, we examined coupled climate tipping elements using the model from Wunderling \textit{et al.} \cite{Wunderling2021}. We focused here on two the case of two generic, mutually reinforcing tipping elements with similar temperature thresholds but different tipping timescales in a scenario when the global mean temperature just passed these thresholds. Our analysis revealed the existence of three ghost states in this model, two of which form a ghost channel, detectable for the whole range of global warming considered here (i.e., up to $\Delta GMT = 3$°C). We further found a striking influence of the tipping timescales on the system's ghost dynamics, which included the change of one of the ghosts dimension by altering the local flow and eigenvalue distribution as well as a restructuring of the ghost connections. Since neither of these transition involved a collision between ghosts, the observed transitions are distinct from the ghost SN-bifurcation observed in the coupled theta neuron model. However, further work is required to better understand the nature of the timescale-dependent transitions. Since long transients due to ghosts can keep a system in a desirable area of phase space even when the equilibrium is lost, a more detailed investigation of the ghost dynamics in a higher dimensional version of this model for realistic parameters, real-world network topologies and policy scenarios is currently underway [Nandan \textit{et al.}, in preparation].

In the final example, we applied our algorithms to GRNs under the assumption that each gene is poised at criticality, i.e., close to a SN-bifurcation. Our analysis showed that embedded in the dynamics of such GRNs can be a complex ghost network in phase space with many higher-dimensional ghosts and a topology distinct from the GRN. These ghost networks can further be multi-stable in the sense of coexisting ghost structures, such that some initial conditions lead to visitation, e.g., of a ghost channel, whereas others end up on a ghost cycle. While the assumption of each gene being poised at criticality is perhaps too strong to be representative of real GRNs, there is a growing body of empirical evidence for ghost dynamics in GRNs \cite{RodrguezMaroto2025,Mercadal2026} and intracellular signaling networks \cite{Stanoev2018}. While the role of ghost dynamics in cell biology is barely explored, its unique properties have been used to explain processing of time-varying stimuli\cite{Stanoev_2020,Koch2024_rev} and put forward as robust mechanisms for cell population control\cite{Simons2025}, regulation of developmental transitions \cite{RodrguezMaroto2025}, and cell-type specification and differentiation \cite{Mercadal2026,Farjami_2021}. Moreover, our example demonstrates that complex ghost dynamics such as described in \cite{Farjami_2021} can easily emerge in random networks and do not require artificial network structures. Exploring complex and higher-dimensional ghost dynamics in GRNs and other types of intracellular regulatory networks thus has the potential to provide new perspectives on various phenomena in cell biology.

Based on these results, we believe the framework and algorithms presented here will open many interesting avenues for further research ranging from physical sciences to living systems.

\section{Methods}
All models in this study were implemented in Python and simulations were performed using the {\tt RK45} integration scheme as implemented in SciPy's {\tt solve$\_$ivp} function \cite{scipyref}. Numerical bifurcation analysis of the coupled theta neurons and tipping element models was performed using XPPAUT. A detailed description of the GhostID algorithm for identifying ghosts from trajectories is given in \autoref{subsec:algorithm} and Supplementary \autoref{sfig:algorithm}. Algorithms for phase space sampling, continuing identified ghosts across parameters and identifying composite ghost structures such as ghost channels and ghost cycles are described in Supplementary section \ref{subsec:PyGhostIDpackage}.

\subsection*{Author contributions}
\small{D.K.: conceptualization, formal analysis, numerical simulations, writing—original draft. A.P.N.: formal analysis, writing—review \& editing.}

\subsection*{Competing Interests}
\small{The authors declare no competing interests.}

\subsection*{Use of generative AI statement}
\small{Generative AI has been used for specific coding tasks and to check the validity of the introduced definitions and theorems. A detailed overview of its use can be found in Supplementary \autoref{subsec:genAIuse}.}

\subsection*{Acknowledgments}
\small{We would like to thank Gayathri Ramesan for helpful discussions during the early stages of this work and Aneta Koseska and Peter Ashwin for their feedback on this manuscript.}

\subsection*{Data availability}
\small{All codes for reproducing the data and figures from this study are available at \url{https://github.com/KochLabCode/PyGhostID}. PyGhostID is available at \url{https://pypi.org/project/PyGhostID/}.}


\printbibliography
\end{refsection}

\newpage 
\begin{refsection}[references_clean.bib]
\thispagestyle{plain}
\begin{center}
\textbf{\Large  - Supplementary Information -}\\
\vspace{0.3cm}
\textbf{\Large Characterization and identification of ghosts and their composite structures in dynamical systems}

\vspace{0.5cm}

Daniel Koch$^{1,2}$, \ Akhilesh P. Nandan$^{3}$

\vspace{0.3cm}

{\small $^{1}$Department of Pharmacology and Therapeutics, University of Manitoba\\
$^{2}$Institute of Cardiovascular Sciences, St. Boniface Hospital Research\\
$^{3}$European Molecular Biology Laboratory, Barcelona, Spain}
\end{center}
\setcounter{section}{0}
\setcounter{subsection}{0}
\renewcommand{\theHsection}{SM.\thesection}
\renewcommand{\theHsubsection}{SM.\thesubsection}
\renewcommand{\theHsubsubsection}{SM.\thesubsubsection}

\startcontents[SM]

\printcontents[SM]{}{1}{%
  \noindent\rule{\linewidth}{0.4pt}\par\smallskip
  \subsection*{Contents}%
}
\noindent\rule{\linewidth}{0.4pt}\par\bigskip

\section{Supplementary Sections}

\subsection{Identifying ghosts via GhostID as implemented in PyGhostID}\label{subsec:ghostIDparas}

The precise steps of the GhostID algorithm are as follows: first, the kinetic-like scalar $Q$ is calculated along a given trajectory to identify local minima as a criterion that the trajectory has entered a region with (relatively) slow dynamics (\autoref{sfig:algorithm}a, step 1-2). Next, around each local minimum a trajectory segment is defined as all points of the trajectory within an $\varepsilon$-neighborhood of the current $Q$-minimum that are part of the same continuous sequence of time steps as the minimum (\autoref{sfig:algorithm}a, step 3). To account for the non-invariance of potential ghosts, the algorithm checks if the trajectory leaves the $\varepsilon$-neighborhood of the current $Q$-minimum again before calculating the eigenvalues along the trajectory segment. Eigenvalues are then evaluated for monotone crossings from negative to positive along the segment (\autoref{sfig:algorithm}a, step 4). If there is at least one eigenvalue that changes sign, the algorithm is records the detection of a ghost with the total number of eigenvalues changing sign being taken as the dimension of the ghost. In the last step, the ghost is compared to previous ghosts recorded by the algorithm (if any) using its position in phase space and given an id based on first or repeated recording (\autoref{sfig:algorithm}a, step 5). The algorithm then repeats steps 3-5 until all $Q$-minima along the trajectory have been processed and returns the sequence of identified ghosts. Each ghost recorded by GhostID is stored as Python dictionary with 9 data fields that can be accessed later for further analysis by the user (\autoref{sfig:algorithm}b). Although our algorithm can detect all ghosts of type $j,k$ saddle-nodes, attracting ghosts naturally will be detected with a greater likelihood than non-attracting ghosts. To distinguish between both, the user can access all eigenvalues at $\textnormal{x}_{\textnormal{Q}_{\textnormal{min}}}$ using the dictionary field {\tt eigenvalues$\_$qmin} and evaluate the signs of all eigenvalues $|\lambda_i|\gg0$. In addition to the direct identification of ghosts via eigenvalue crossings, the algorithm features an alternative, indirect method by evaluating the gradients of the eigenvalues which is more sensitive to detect ghosts along a trajectory but can lead to false positive ghost identifications and thus needs to be explicitly enabled by the user (cf. Supplementary \autoref{sfig:falseGhost}).

As user input the algorithm requires a trajectory, the ODE system, its parameter values and the step size of the integration. The main hyper-parameters of the algorithm are the size $\varepsilon$ of the environment around the $Q$-minima and the distance $\delta$ in phase space between any two ghosts identified by the algorithm above which it considers them to be distinct. The package builds on existing scientific computing libraries within the Python ecosystem, including JAX \cite{jax2018github}, SciPy \cite{scipyref}, and NumPy \cite{numpyref}, making it computationally efficient, versatile, and highly customizable by the user. For instance, $Q$-minima along trajectories are identified via SciPy's {\tt find$\_$peaks} function on $pQ(t) := -\log(Q(t))$ which the user is able to control via all available hyper-parameters of {\tt find$\_$peaks}. To aid the choice of hyper-parameters, our Python implementation features two types of control plots: $pQ(t)$ showing identified $Q$-minima and evaluation of the ghost criteria along trajectory segments (including eigenvalues). 

GhostID as implemented in PyGhostID is used via the following function call:

\begin{lstlisting}[language=Python]
ghostID(model, params, dt, trajectory, epsilon_gid = 0.05, delta_gid=0.05, **kwargs)
\end{lstlisting}

where {\tt model} is the Python function describing the system dynamics, {\tt parameters} are the model parameters to be given as argument to {\tt model}, {\tt dt} is the stepsize and {\tt trajectory} is a trajectory simulated by the system. It returns {\tt ghostSeq}, a Python list of identified ghost states (Python dictionary) and, if {\tt return$\_$ctrl$\_$figs = True}, the figures for the control plots requested by the user. The hyper-parameters of the algorithm are as follows:

\begin{itemize}
    \item {\tt epsilon\textunderscore gid} (float): Radius of the $\varepsilon$-sphere around $Q$-minima determining the trajectory segments along which eigenvalues are evaluated. For many models considered here we found values in the range of $0.01$ to $0.1$ a reasonable choice, hence the default value is set to $0.05$. However, suitable values strongly depend on the model and the typical ranges of the phase space variables. Generally, the value should be chosen big enough such that trajectory segments consist of enough points to allow a reliable identification of the eigenvalue spectrum along the segment, but small enough so that the eigenvalues are still representative of the local phase-space topology around potential ghosts. A good strategy for choosing {\tt epsilon\textunderscore Qmin} is to start with small values and increase it until eigenvalue control plots look reasonably smooth.
    \item {\tt delta\textunderscore gid} (kwarg, float): Distance in phase space between two ghosts identified by GhostID above they are considered distinct and given different identifiers. Default value is $0.1$.
    \item {\tt peak\textunderscore kwargs} (kwarg, dictionary): Can contain any kwarg for SciPy's {\tt find\textunderscore peaks} function to improve the detection of $Q$-minima from peaks in the $pQ$-timeseries.
    \item {\tt evLimit} (kwarg, float): Default value is 0. Enables indirect method of identifying ghosts if {\tt evLimit}$>0$. A trajectory segment is considered to pass nearby a ghost if the following holds true: the absolute value of the median of the eigenvalues along the trajectory segment is below {\tt evLimit}, the linear fit of the eigenvalues along the segment has an $R^2\geq0.99$ and a slope that lies within the range given by the kwarg {\tt slopeLimits} (two-element array of floats, default is $[0,\infty]$).
    \item {\tt batchModel} (kwarg, function): Vectorized version of the model able to handle batch inputs. Either manually coded model function or made via {\tt make$\_$batch$\_$model}. While the algorithm calls {\tt make$\_$batch$\_$model} itself at its initialization, providing a pre-calculated batch model can be useful to improve performance when {\tt ghostID} is called many times using the same model.
    If at least one control plot is enabled, control plots will either been shown inline or returned if {\tt return\textunderscore ctrl\textunderscore figs} is set to {\tt True}
\end{itemize}

Eigenvalues are calculated by {\tt ghostID} using the {\tt jax.numpy.linalg.eigvals} function. However, there is no guarantee that the indexing of the eigenvalues along a trajectory segment remains the same for consecutive points along the segment, i.e. $\lambda_1$ at $t$ might be $\lambda_2$ at step $t+1$, thereby potentially interferring with the identification of ghosts. While we found the indexing to remain the same for the majority of cases, we also found cases in which eigenvalue indices were permuted. {\tt ghostID} has two complementary methods to deal with such cases, outlier removal and sorting of eigenvalues across time, which can be controlled as follows:

\begin{itemize}
\item {\tt ev$\_$outlier$\_$removal} (kwarg, boolean): Removes eigenvalues along a trajectory based on outlier detection, where eigenvalues are removed if, along a sliding window, they fall above $q_3 + k (q_3 - q_1)$ or below $q_1 - k (q_3 - q_1)$, where $k$ is a positive float, $q_1$ and $q_3$ the 25th and 75th percentile, respectively. Default value is {\tt False}.
\item {\tt ev$\_$outlier$\_$removal$\_$k} (kwarg, float): Determines the size of the range outside which eigenvalues are considered outliers. Default value is $1.5$.
\item {\tt ev$\_$outlier$\_$removal$\_$ws} (kwarg, float): Determines the size of the sliding windows in which eigenvalues are evaluated for outliers. Default value is $7$.
\item {\tt eigval$\_$NN$\_$sorting} (kwarg, boolean): Sorts eigenvalues $\lambda_i(t)$ across time for a given index $i$ by making a linear prediction $\lambda_p$ of the next real value of $Re(\lambda_i(t+1))$ and assigns $\lambda_i(t+1) = \lambda_j(t+1)$ for $1\leq j \leq n$ for which $|\lambda_p-\lambda_j(t+1)|$ is minimal . Default value is {\tt False}.
\end{itemize}

{\tt ghostID} also features several control outputs that are helpful for trouble-shooting and selecting hyper-parameters:

\begin{itemize}
\item {\tt display\textunderscore warnings} (kwarg, boolean): show/hide warning messages from GhostID.
\item {\tt ctrlOutputs}(kwarg, dictionary): Several keys can be used to plot the algorithm's two core quantities, $Q$- and eigenvalues, and customize the plots.
\begin{itemize}
    \item {\tt ctrl$\_$qplot} (boolean): enables plot of $pQ$-values and $Q$-minima along trajectory if set to {\tt True}.
    \item {\tt qplot$\_$xscale} (string): set scale of x-axis to {\tt "linear"} (default) or {\tt "log"}.
    \item {\tt qplot$\_$yscale} (string): set scale of y-axis to {\tt "linear"} (default) or {\tt "log"}.
    \item {\tt ctrl$\_$evplot} (boolean): if set to {\tt True}, a plot of eigenvalues along each trajectory segment around identified $Q$-minima is shown including information about the evaluation criteria for ghosts which are listed in the plot's heading.
    \item {\tt evplot$\_$xscale} (string): set scale of x-axis to {\tt "linear"} (default) or {\tt "log"}.
    \item {\tt evplot$\_$yscale} (string): set scale of y-axis to {\tt "linear"} (default) or {\tt "log"}.
\end{itemize}
   \item {\tt return\textunderscore ctrl\textunderscore figs}(kwarg, boolean): Returns control plots for manual customization of plot settings if set to {\tt True}. Default value is {\tt False}.
\end{itemize}

\subsection{Other core functionalities of PyGhostID}\label{subsec:PyGhostIDpackage}

\subsubsection{ghostID\textunderscore phaseSpaceSample}\label{subsec:ghostid_pSS}

The function {\tt ghostID\textunderscore phaseSpaceSample} simulates many trajectories within a specified range of the phase space from which initial conditions are chosen by Latin hypercube sampling and trajectories are evaluated by {\tt ghostID}. The function implements an adaptive parallel-processing routine to speed up simulation by using threads when run in Jupyter or processes for full CPU utilization when run as a standalone script. The function is called via:

\begin{lstlisting}[language=Python]
ghostID_phaseSpaceSample(model, model_params, t_start, t_end, dt, state_ranges, n_samples=50, method='RK45', rtol=1.e-3, atol=1.e-6, n_workers=None, **kwargs):
\end{lstlisting}

It returns a list containing a {\tt ghostSeq} (see subsection \ref{subsec:ghostIDparas}) for each trajectory. The arguments are as follows:  {\tt model} is the Python function describing the system dynamics, {\tt model$\_$parameters} are the model parameters to be given as argument to {\tt model}, {\tt t$\_$start} and {\tt t$\_$end} are beginning and endpoints of each trajectory, {\tt dt} is the stepsize,  {\tt state$\_$ranges} is a list with $n$ elements ($n$ being the system's dimension), each of which is a tuple specifying the lower and upper boundaries of the respective coordinates of the phase space region to be sampled, {\tt n$\_$samples} the number of trajectories to be simulated and analyzed for ghosts. {\tt method, rtol} and {\tt atol} are arguments for SciPy's {\tt solve\textunderscore ivp} routine used for simulating the model and control the method (default is `{\tt RK45}'), the relative and absolute tolerances for integration, respectively. {\tt n\textunderscore workers} is the number of cores used for parallel processing which chosen automatically if {\tt n\textunderscore worker = None}. Furthermore, {\tt ghostID\textunderscore phaseSpaceSample} accepts all keyword arguments that apply to {\tt ghostID}. Note, however, that showing control plots likely will not work when enabling parallel processing. In addition, there two more {\tt kwargs}:
\begin{itemize}
    \item {\tt delta\textunderscore unify} (kwarg, float): Determines, similar to {\tt delta\textunderscore gid}, the distance in phase space between two identified ghosts above which they are considered distinct and given different identifiers but \textit{across all trajectories of a phase space sample}. Default value is $0.1$.
    \item {\tt seed} (kwarg, integer): Choosing a specific seed enables exact reproducibility of the (pseudo-)random phase space sample. Default value is {\tt None}.
\end{itemize}

\subsubsection{ghost\textunderscore connections}

The function {\tt ghost\textunderscore connections} takes a list of {\tt ghostSeq}s (see subsection \ref{subsec:ghostIDparas}), either generated by 
{\tt ghostID\textunderscore phaseSpaceSample} or by collecting different runs of {\tt ghostID} (after unifying identifiers) in a list, as input and returns an adjacency matrix $(a_{ij})_{1\leq i,j \leq n}$ describing the connections between ghosts. The underlying algorithm is simple: for any two consecutive ghosts $G_i$ and $G_j$ within a {\tt ghostSeq}, set $a_{i,j} = 1$. The function is called without any hyperparameters via:

\begin{lstlisting}[language=Python]
ghost_connections(gSeqs)
\end{lstlisting}

\subsubsection{track\textunderscore ghost\textunderscore branch}

The function {\tt track\textunderscore ghost\textunderscore branch} allows the user follow a specific ghost as a system parameter is varied. The algorithm consists of the following steps:

\begin{enumerate}
    \item Start with position $x_g\in \mathbb{R}^n$ of a ghost $G$ at parameters $\rho\in\mathbb{R}^m$.
    \item Change $\rho_i$ by $\Delta \rho$.
    \item Within a small region around $x_g$, find the next closest $Q$ minimum $x_{Q_{min}}$ in phase space.
    \item Simulate a trajectory $T$ starting from a specified distance of $x_{Q_{min}}$.
    \item Apply GhostID to $T$ to find a ghost. If a ghost is found, save the ghost and repeat the process until no more ghost is found or the maximum number of repeats is reached.
\end{enumerate}

The function is called as follows:

\begin{lstlisting}[language=Python]
track_ghost_branch(ghost, model, model_params, par_nr, par_steps, dpar, t_end, dt, delta=0.5,icStep=0.1, mode="first",epsilon_gid=0.1,solve_ivp_method='RK45', rtol=1.e-3, atol=1.e-6,**kwargs)
\end{lstlisting}

It returns the positions in phase space of the ghosts found ({\tt ghostPositions}), the sequence of parameters at which these ghosts were found ({\tt parSeq}), and the full list of ghosts found at the different parameter values ({\tt ghostSeq\textunderscore p}) to enable extraction of other information about the ghosts in addition to the positions in phase space. If enabled, it also returns the control plots for GhostID.

The arguments and hyperparameters are as follows:

\begin{itemize}
    \item {\tt ghost} (python dictionary): Ghost to be tracked identified either by {\tt ghostID} or {\tt ghostID\textunderscore phaseSpaceSample}.
    \item {\tt model} (function): Python function describing the system dynamics.
    \item {\tt model\textunderscore params}: parameters for {\tt model}.
    \item {\tt par\textunderscore nr} (integer): index of parameter to be varied during tracking.
    \item {\tt par\textunderscore steps} (integer): number of times a parameter is updated during a sweep.
    \item {\tt dpar} (float): size of parameter increment (if {\tt dpar} is positive) or decrement (if {\tt dpar} is negative) per iteration.
    \item {\tt t\textunderscore end} (float): length of trajectory at each iteration step.
    \item {\tt dt} (float): step-size of numerical integration.
    \item {\tt delta} (float): size of region around $x_g$ in which to search for $x_{Q_{min}}$. Default value is $0.5$.
    \item {\tt icStep} (float): distance from $x_{Q_{min}}$ at which to initialize a trajectory which will be analyzed by GhostID. Default value is $0.1$.
    \item {\tt mode} (string): If {\tt mode = "first"} (default), the algorithm takes the first ghost of potentially multiple ghosts identified along the trajectory. If {\tt mode = "closest"}, the algorithm takes the ghost closest in phase space to $x_g$ of potentially multiple ghosts identified along the trajectory.
    \item {\tt epsilon\textunderscore gid} (float): see section \ref{subsec:ghostIDparas}.
    \item {\tt solve\textunderscore ivp\textunderscore method, rtol, atol}: arguments for numerical integration, see section \ref{subsec:ghostid_pSS}.
    \item {\tt distQminThr} (kwarg, float): constraint mandating that any ghost identified must not have a distance greater than {\tt distQminThr} from $x_{Q_{min}}$. Default value is $\infty$.
\end{itemize}

Additionally, the function can receive any {\tt kwarg} that applies to the {\tt ghostID} function (see section \ref{subsec:ghostIDparas}) and any {\tt kwarg} that applies to {\tt find\textunderscore local\textunderscore Qminimum} (see section \ref{subsec:helperFunctions}).

\subsection{Helper functions of PyGhostID}\label{subsec:helperFunctions}

In addition to the core functionalities, PyGhostID features several helper function that can be useful in different contexts:\\

\noindent \textbf{unify\textunderscore IDs}

Unifies IDs of ghosts across trajectories. Call via:
\begin{lstlisting}[language=Python]
unify_IDs(gSeqs, delta_unify=0.1, update=True)
\end{lstlisting}

Arguments:
\begin{itemize}
    \item {\tt seqs}: list of {\tt ghostSeq}s from several trajectories generated, e.g., by collecting different runs of {\tt ghostID} in a list.
    \item {\tt delta\textunderscore unify} (float): distance in phase space between two identified ghosts above which they are considered distinct and given different identifiers. Any two ghosts from different trajectories are given the same id if they are separated by a distance less than {\tt delta\textunderscore unify}. Default value is $0.1$. 
    \item {\tt update} (boolean). If {\tt update = True} (default), the algorithm also synchronizes $Q$-value and position based on the ghost with lower $Q$-value and dimension based on the ghost with the higher dimension.    
\end{itemize}

\noindent \textbf{unique\textunderscore ghosts}

Returns the number of unique ghosts based on the ids of the ghosts in a {\tt ghostSeq}. Call via:
\begin{lstlisting}[language=Python]
unique_ghosts(gSeq)
\end{lstlisting}

\noindent \textbf{make\textunderscore batch \textunderscore model}

Takes a model function and its parameters to create a vectorized version of the model to enable batch processing. Increases performance when running GhostID on many trajectories. Call via:
\begin{lstlisting}[language=Python]
make_batch_model(model, params)
\end{lstlisting}

\noindent \textbf{find\textunderscore local\textunderscore Qminimum}

Function searches for a $Q$-minimum within an neighborhood of radius $\delta$ around a given $x\in\mathbb{R}^n$ in phase space using either a sampling approach or global optimization method, followed optionally by local optimization. The function is called as follows:

\begin{lstlisting}[language=Python]
find_local_Qminimum(model,x0,model_params,delta,*,global_method="lhs",local_method="L-BFGS-B",global_options=None,local_options=None,verbose=False)
\end{lstlisting}

The arguments and hyperparameter are as follows:

\begin{itemize}
    \item {\tt model} (function): Python function describing the system dynamics.
    \item {\tt x0} (array): center of region in which to search for $Q$-minimum.
    \item {\tt model\textunderscore parameters}: parameters for {\tt model}.
    \item {\tt delta} (float): radius of neighborhood around {\tt x0} in which to search for $Q$-minimum.
    \item {\tt global method (string)}: 
    \begin{itemize}
        \item if {\tt global method = "lhs"} (default), the algorithm takes a Latin hypercube sample from the $\delta$-neighborhood around {\tt x0}, evaluates $Q$ at each point from the sample and returns a specified number of points with the lowest $Q$-values to be taken as starting point for local optimization.
        \item if {\tt global method = "differential\textunderscore evolution"}: uses differential evolution algorithm from {\tt scipy.optimize} to find a $Q$-minimum.
        \item if {\tt global method = "dual\textunderscore annealing"}: uses dual annealing algorithm from {\tt scipy.optimize} to find a $Q$-minimum.
         \item if {\tt global method = "basin\textunderscore hopping"}: uses basin hopping algorithm from {\tt scipy.optimize} to find a $Q$-minimum.
    \end{itemize}
    \item {\tt local\textunderscore method} (string): all methods that are available in {\tt scipy.optimize.minimize}, default is {\tt "L-BFGS-B"}.
    \item {\tt global\textunderscore options} (dict): use all allowed kwargs for global methods from {\tt scipy.optimize} or, if {\tt global method = "lhs"}, the following kwargs to control method:
    \begin{itemize}
        \item {\tt n\textunderscore samples} (integer): size of the LHS-sample. If {\tt None}, the sample size will be automatically calculated as {\tt min(2000, max(200, 20 * dim))}, where {\tt dim} is the dimension of the model.
        \item {\tt k\textunderscore seeds} (integer): number of points with the lowest $Q$-values to be taken as starting point for local optimization. f {\tt None}, the sample size will be automatically calculated as {\tt min(5, max(2, int(sqrt(dim))))}, where {\tt dim} is the dimension of the model.
        \item {\tt seed} (integer): Choosing a specific seed enables exact reproducibility of the (pseudo-)random LHS sample. Default value is {\tt None}.
    \end{itemize}
    \item {\tt local\textunderscore options} (dict): use all allowed kwargs for local methods from {\tt scipy.optimize.minimize}.
    \item {\tt verbose} (boolean): enables control outputs if set to {\tt True}. Default is {\tt False}.
\end{itemize}

\noindent \textbf{qOnGrid}

Function calculates the $Q$-values on a phase space grid. The function is called as follows:

\begin{lstlisting}[language=Python]
qOnGrid(model, model_params, coords=None, n_points=50, ranges=None, overrides=None, indexing="ij", jit=False)
\end{lstlisting}

The arguments and hyperparameters are as follows:

\begin{itemize}
    \item {\tt model} (function): Python function describing the system dynamics.
    \item {\tt model\textunderscore parameters}: parameters for {\tt model}.
    \item {\tt coords} (list of 1d arrays): the phase space grid on which to evaluate $Q$-values. If {\tt None}, the grid is calculated automatically according to following parameters:
    \begin{itemize}
        \item {\tt n\textunderscore points} (integer / list of integers): number of grid points along each dimension specified for all dimensions at once or individually via list of integers. Default is $50$.
        \item {\tt ranges} (tuple / list of tuples). Upper and lower bounds of grid points along each dimension specified for all dimensions at once or individually via list of tuples. Default is $(-2,2)$.
        \item {\tt overrides} (dict): dictionary to override {\tt n\textunderscore points} and/or {\tt ranges} for specific individual axes, without changing the other ones, e.g., {\tt {1: {"n": 100, "range": (-5.0, 5.0)}}}.
        \item {\tt indexing} (string): indexing used for {\tt jnp.meshgrid} function from jax. Default is {\tt "ij"}.
        \item {\tt jit} (boolean): enables jit to speed up performance if {\tt True}. Default is {\tt False}.
    \end{itemize}

\end{itemize}

\noindent \textbf{draw\textunderscore network}

Takes adjacency matrix (e.g., generated by {\tt ghost\textunderscore connections}) to draw network via networkX package. Call function as follows:

\begin{lstlisting}[language=Python]
draw_network(adj_matrix, nodeCols, nlbls, layout="fdp", graphviz_args=None, layout_kwargs=None, rankdir="TB", node_size=1800, label_font_size=16.5, font="Arial")
\end{lstlisting}

The arguments and hyperparameter are as follows:

\begin{itemize}
\item  {\tt adj\textunderscore matrix} (2d array): adjacency matrix to be drawn as network. Entries of $1$ will lead to black edges, entries of $-1$ will lead to red edges.
\item {\tt nodeCols} (list with colors): Colors of nodes in network.
\item {\tt nlbls} (list of strings): Node labels.
\item {\tt layout} (layout): Graphviz layout ({\tt `fdp', `dot'/`hierarchical', `neato', `sfdp', `circo'}) or any {\tt nx.*\textunderscore layout} function from networkX. Default is Graphviz {\tt `fdp'}.
\item  {\tt graphviz\textunderscore args} (string): arguments for Graphviz layouts. If {\tt None}, default string is {\tt f"-Grankdir={rankdir} -Nwidth=350 -Nheight=350 -Nfixedsize=true -Goverlap=scale -Gnodesep=5000 -Granksep=200 -Nshape=oval -Nfontsize=14 -Econstraint=true"}.
\item {\tt layout\textunderscore kwargs} (dict): kwargs for {\tt nx.*\textunderscore layout} function from networkX.
\item {\tt rankdir} (string): Direction of node ordering for Graphiviz layout {\tt `dot'}. Default is {\tt `TB'}, from top to bottom.
\item {\tt node\textunderscore size} (float): Size of nodes in network. Default is $1800$.
\item {\tt label\textunderscore font\textunderscore size} (float): Size of node labels. Default is $16.5$.
\item {\tt font} (string): Font type for node labels. Default is {\tt "Arial"}.
\end{itemize}

\subsection{Further numerical validation of GhostID}\label{subsection:numVal}

\subsubsection{Additional models with long transients due to ghosts}\label{subsection:numVal_ghosts}

\textbf{EGF-receptor model}\\
A minimal model EGR-receptor activity that captures experimentally the observed ghost dynamics in \cite{Stanoev2018} can be given by the equations \cite{Stanoev_2020}:
\begin{align*}
\dot R_a &= k_R(R_i(\alpha_1R_i + \alpha_2 R_a + \alpha_3 LR_a)-\hat{\gamma}_{DNF}P_{DNF,a}R_a),\\
\dot P_{DNF,a} &=k_1(P_{DNF,i}-k_{2/1}P_{DNF,a}-\hat{\beta}_{DNF}P_{DNF,a}(R_a+LR_a)),
\end{align*}

where $R_a$ and $R_i=1-R_a$ are the active and inactive forms of the ligandless EGF-receptor, $LR_a$ is the active ligand-bound form of the receptor which is treated as input parameter, and $P_{DNF,a}$ and $P_{DNF,i} = 1-P_{DNF,a}$ are active and inactive form of the phosphatases inactivating the receptor. Parameters $\alpha_{1-3}$ are receptor activation rate constants, $k_R$ and $k_1$ are kinetic constants that do not influence the steady-state values of the system, $\hat{\gamma}_{DNF}$ the specific reactivity of the enzyme towards the receptor, and $\hat{\beta}_{DNF}$ and $k_{2/1}$ are the receptor-induced regulation rate constant and the inactivation/activation constant ratio of $P_{DNF}$, respectively. As shown in Supplementary \autoref{sfig:numval_ghosts}a, GhostID correctly identified the ghost in this system.
\\ \\

\textbf{Charge density wave model}\\
We further consider the following model for delayed switching of charge-density waves shown to have ghost dynamics in \cite{Strogatz1989}:
\begin{align*}
\dot \psi &= E + \frac{1}{2}(\sin(\psi+\phi)+sin(\psi-\phi)),\\
\dot \phi &= \frac{1}{1+2\gamma}(-2K\psi\phi + \frac{1}{2}(\sin(\psi+\phi)+sin(\psi-\phi))),
\end{align*}

where $\psi$ and $\phi$ are phase differences between the two coupled domains and the current carried by the CDW is proportional to $\dot\psi$, $E$ is proportional to the applied electric field, $\gamma$ is the viscous and $K$ the elastic coupling strength, respectively. The system has several ghosts that are reliably identified by GhostID (Supplementary \autoref{sfig:numval_ghosts}b).

\textbf{Gene regulatory network model}\\
The following model by Farjami \textit{et al}. \cite{Farjami_2021} has been proposed to explain differentiation of stem cells by cycling between transient but metastable states as an alternative model to the classical landscape picture by C. Waddington \cite{Kelsh2021}:

\begin{align*}
    \dot{x}_1 &= b_1 + \frac{g_1}{(1 + \alpha_1 x_2^h)(1 + \beta_1 x_3^h)} - d_1 x_1 \\
    \dot{x}_2 &= b_2 + \frac{g_2}{(1 + \alpha_2 x_3^h)(1 + \beta_2 x_1^h)} - d_2 x_2 \nonumber\\
    \dot{x}_3 &= b_3 + \frac{g_3}{(1 + \alpha_3 x_1^h)(1 + \beta_3 x_2^h)} - d_3 x_3, \nonumber
\end{align*}

where $x_i$ are the gene products (proteins) that act as transcription factors, $b_i$ and $g_i$ are background and maximal gene expression rates, $\alpha_i$ and $\beta_i$ are association constants of the transcription factors to DNA, $h_{i}$ are Hill coefficients describing cooperativity, and $d_i$ are the degradation rates. Given sufficient cooperativity and assuming for simplicity that all gene expression rates share the same parameters, the system starts oscillating via a Hopf-bifurcation as $g_i$ increases before oscillations are terminated again by three simultaneously occuring SNIC bifurcations upon further increase of $g_i$ \cite{Farjami_2021}. At $g_i$ close to the SNIC bifurcations, the system transitions periodically between three ghost states, forming what we call a ghost cycle \cite{Farjami_2021,Koch_2024}. Applying GhostID to a trajectory from this system close to the SNIC bifurcations, again, correctly identifies three ghosts in phase space, each of which is identified by a single eigenvalue crossing from $<0$ to $>0$ (Supplementary \autoref{sfig:numval_ghosts}c)\\

\subsubsection{Additional model with long transients due to saddles}\label{subsection:numVal_saddles}

Consider the classical May-Leonard cycle of winner-less competition between three species \cite{May1975}:

\begin{align*}
    \dot{N}_1 &= N_1 \left(1 - N_1 - \alpha N_2 - \beta N_3\right), \\
    \dot{N}_2 &= N_2 \left(1 - \beta N_1 - N_2 - \alpha N_3\right), \nonumber \\
    \dot{N}_3 &= N_3 \left(1 - \alpha N_1 - \beta N_2 - N_3\right), \nonumber
\end{align*}

where $N_i$, $\alpha$ and $\beta$ are the non-dimensionalized species abundances and competition parameters, respectively. For the parameter regime shown in Supplementary \autoref{sfig:numval_HC} (left), the system exhibits a stable heteroclinic cycle in which deterministic trajectories get closer and closer to the saddles of the heteroclinic connections with each iteration, leading to increasingly long transients. None of the eigenvalues along the trajectory segments of the corresponding $Q$-minima changes sign, hence GhostID does not identify any ghosts along the trajectory (Supplementary \autoref{sfig:numval_HC} middle and right).

\subsubsection{Additional models with long transients due to slow-fast dynamics}\label{subsection:numVal_slowFast}

\textbf{2D slow-fast toy model}\\
Consider the following model from \cite{Kuehn_2015}:
\begin{align*}
    \dot{x} &= \epsilon - x, \\
    \dot{y} &= \epsilon y^2. \nonumber
\end{align*}

Since the slow manifold of this system is divided by a saddle-node into two halves, we apply GhostID on two trajectories, one that approaches the saddle-node and one that is repelled by it (Supplementary \autoref{sfig:numval_SF}a, left). While in the first case, we see only a monotonic decrease in $Q(t)$, the second trajectory shows a pronounced $Q$-minimum underlying a long transient (Supplementary \autoref{sfig:numval_SF}a, middle). Yet, none of the eigenvalues changes sign, hence the long transient on the slow manifold is not identified as a ghost (Supplementary \autoref{sfig:numval_SF}a, right).\\

\textbf{Michaelis-Menten model of enyzme kinetics}\\
Considering the classical Michaelis-Menten mechanism of enzymatic catalysis as a singular perturbation problem \cite{Heineken1967}, one obtains explicit time-scale separation that with the suitable non-dimensionalization \cite{roussel2005singular} leads to:

\begin{align*}
\dot s &= \beta c (1 - \alpha) - s (1 - \alpha c), \\
\dot c &= \frac{1}{\mu}(s (1 - \alpha c) - c (1 - \alpha)),
\end{align*}

where $\alpha$ and $\beta$ are dimensionless parameters defined in terms of the original rate constants and the Michaelis-constant, and $\mu$ is the timescale parameter. As shown in \autoref{sfig:numval_SF}b, GhostID correctly does not identify long transients in this model as ghosts.
\\ \\
\textbf{Slow-fast dynamics in model from Hastings \textit{et al.} 2018}\\
The model described by equation \ref{eq:Hastings} from \cite{Hastings_2018} used to test GhostID on saddle crawl-bys also features a parameter regime with classical slow-fast dynamics. As shown in Supplementary \autoref{sfig:numval_SF}c, GhostID correctly does not identify long transients at this regime as ghosts.
\\ \\
\textbf{van der Pol oscillator}\\
As a last example, we consider the van der Pol model \cite{van_der_Pol_1928}, a classical slow-fast system given by the equations \cite{Kuehn_2015}:

\begin{align}\label{eq:Vdp}
    \nonumber \dot{x} &= \frac{1}{\epsilon}(x-\frac{x^3}{3}+y)  \\ 
    \dot{y} &= -x,
\end{align}

\noindent where $\epsilon$ is the timescale parameter.  As shown in Supplementary \autoref{sfig:numval_SF}d, GhostID correctly does not identify long transients in this model as ghosts.

\subsection{Eigenvalues along trajectories in neighborhoods of hyperbolic fixed points}\label{subsec:theorem_vicinityHypFP}

\begin{theorem}[Eigenvalues along trajectories in neighborhoods of hyperbolic fixed points]
\label{theorem_vicinityHypFP}
Let $\dot{x} = \textbf{f}(x)$, $\textbf{f}:\mathbb{R}^n\to\mathbb{R}^n$, be a smooth, autonomous dynamical system with a hyperbolic fixed point $x^*$. Then there exists a $\delta>0$ such that the real parts of the instantaneous eigenvalues at the points along any trajectory within $\mathcal{U}_\delta(x^*)$ do not change sign.
\end{theorem}

\begin{proof} Since $\textbf{f}$ is smooth, $\partial f_i/ \partial x_j$ exists and is continuous for all $1\leq i,j\leq n$. Recall that the eigenvalues of a square matrix (real or complex) depend continuously on its entries (cf. Theorem 2.4.9.2. from \cite{HornJohnson2013}), hence the eigenvalues $\lambda_i(x)$ of the Jacobian, \[D_f(x)=  \begin{bmatrix}
\frac{\partial f_1}{\partial x_1} &  \cdots & \frac{\partial f_1}{\partial x_n} \\
\vdots & \ddots & \vdots \\
\frac{\partial f_n}{\partial x_1}& \cdots & \frac{\partial f_n}{\partial x_n}
\end{bmatrix},
\] depend continuously on $x$. As $x^*$ is hyperbolic, $Re(\lambda_i)\neq0$. For each $\varepsilon_i>0$, by continuity of $\lambda_i(x)$, there is a $\delta_i>0$, such that for $x \in \mathcal{U}_{\delta_i}(x^*)$ we have $|\lambda_i(x^*)-\lambda_i(x)|<\varepsilon_i$. In particular, we can choose $\varepsilon_i$ small enough for $Re(\lambda_i(x))$ to have the same sign $Re(\lambda_i(x^*))$. For $\delta=\min \{\delta_1, \dots, \delta_n\}$, the eigenvalue $\lambda_i(x)$ at any point $x$ of a trajectory $T\subset\mathcal{U}_\delta(x^*)$ has thus the same sign as $\lambda_i(x^*)$ for all $1\leq i \leq n$.
\end{proof}

\subsection{Evaluation of equilibrium in coupled theta neuron model}\label{subsec:theorem_jKSNinThetaNeurons}

\begin{theorem}[Existence of a type $2,0$ saddle-node equilibrium in coupled theta neurons] \label{theorem_jKSNinThetaNeurons}
For $\eta_1 = \eta_2  = 0$ and $\kappa >0$, system \ref{eq:augustson} has a type $2,0$ saddle-node equilibrium at the origin.
\end{theorem}

\begin{proof} 
Denote $\mathbf{\theta} = (\theta_1,\theta_2)$. For $\theta=(0,0)$ we have 

\begin{align*}
\dot{\theta}_1 &= 1 - \cos\theta_1 + \left(1 + \cos\theta_1\right) \left(\eta_1 + \kappa \left(1 - \cos\theta_2\right)\right)\\ 
&= 1 - \cos0 + \left(1 + \cos0\right) \left(0 + \kappa \left(1 - \cos0\right)\right)\\
&= 1 - 1 + \left(1 + 1\right) \left(0 + \kappa \left(1 - 1\right)\right) = 0\\
\end{align*}
and analogously $\dot \theta_2=0$, showing that the origin is an equilibrium point.
We next prove that the origin fulfills (SN1) from definition \ref{def:SN}. The Jacobian matrix of system \ref{eq:augustson} is
$$D_{\theta} f(\theta,(\eta_1,\eta_2,\kappa))=\begin{bmatrix}
\sin\theta_1(1-\eta_1-\kappa(1-\cos\theta_2)) & \kappa\sin\theta_2(1+\cos\theta_1) \\
\kappa\sin\theta_1(1+\cos\theta_2) & \sin\theta_2(1-\eta_2-\kappa(1-\cos\theta_1)) 
\end{bmatrix}.$$ We thus have  
$$D_{\theta} f((0,0),(0,0,\kappa))=\begin{bmatrix}
0 & 0 \\
0 & 0
\end{bmatrix},$$ making the equilibrium at the origin non-hyperbolic with eigenvalues $\lambda_1=\lambda_2=0$. That is, have a geometric multiplicity of $2$ for the eigenvalue $0$. To find the eigenvectors $w\in \mathbb{R}^2$, we need to solve $$(D_{\theta} f((0,0),(0,0,\kappa))-\lambda_i I)v=
\begin{bmatrix}
-\lambda_i & 0 \\
0 & -\lambda_i
\end{bmatrix}
\begin{pmatrix}
v_1 \\
v_2
\end{pmatrix}
= \begin{pmatrix} 0 \\ 0\end{pmatrix}.
$$
Since $\lambda_i=0$, every non-zero vector $v\in\mathbb{R}^2$ fulfills above equation and is thus an eigenvector of $\lambda_i$. Therefore, the center eigenspace is $E^c=\{v: (D_{\theta} f((0,0),(0,0,\kappa))-\lambda_i I) \ v=0\} = \mathbb{R}^2$. Since $\textnormal{dim}(E^c)=2$, we have a geometric multiplicity of $2$, making $\lambda_i$ semi-simple with eigenvalue $v_i$ to the right. Since the Jacobian is symmetric, we can simply set $w_i=v_i^T$ as the eigenvector to the left. Hence (SN1) is satisfied.\\ \\
We next show that (SN2) is fulfilled. Denote $$f((\theta_1,\theta_2),(\eta_1,\eta_2,\kappa)) = \begin{pmatrix}  1 - \cos\theta_1 + \left(1 + \cos\theta_1\right) \left(\eta_1 + \kappa \left(1 - \cos\theta_2\right)\right) \\  1 - \cos\theta_2 + \left(1 + \cos\theta_2\right) \left(\eta_2 + \kappa \left(1 - \cos\theta_1\right)\right)\end{pmatrix}.$$
Then $\frac{\partial f}{\partial \eta_1}((\theta_1,\theta_2),(\eta_1,\eta_2,\kappa)) = \begin{pmatrix}1 + \cos\theta_1 \\ 0\end{pmatrix}$, so $w_i\frac{\partial f}{\partial \eta_1}((0,0),(0,0,\kappa))=w_i\begin{pmatrix}2\\0\end{pmatrix}$. Analogously we have $\frac{\partial f}{\partial \eta_2}((\theta_1,\theta_2),(\eta_1,\eta_2,\kappa)) = \begin{pmatrix}0\\1 + \cos\theta_2\end{pmatrix}$, so $w_i\frac{\partial f}{\partial \eta_2}((0,0),(0,0,\kappa))=w_i\begin{pmatrix}0\\2\end{pmatrix}$. Since $w_i\neq\begin{pmatrix}0 \ 0\end{pmatrix}$, we thus have $w_i\frac{\partial f}{\partial \eta_n}((0,0),(0,0,\kappa))\neq0$ for either $n=1$ or $n=2$, showing that (SN2) is fulfilled.
\\ \\
We now show that (SN3) is fulfilled, too. Denote $f_1=\dot\theta_1$ and $f_2=\dot\theta_2$, giving the the following second order partial derivatives:

\begin{align*}
    \frac{\partial^2 f_1}{\partial^2 \theta_1} &= \cos\theta_1(1-\eta_1-\kappa(1-\cos\theta_2)),\quad
    \frac{\partial^2 f_1}{\partial \theta_1 \partial \theta_2} = -\sin\theta_1\sin\theta_2,\quad
    \frac{\partial^2 f_1}{\partial^2 \theta_2} = \kappa\cos\theta_2(1-\cos\theta_1),\\
    \frac{\partial^2 f_2}{\partial^2 \theta_1} &= \kappa\cos\theta_1(1-\cos\theta_2),\quad
    \frac{\partial^2 f_2}{\partial \theta_1 \partial \theta_2} = -\sin\theta_1\sin\theta_2, \quad
    \frac{\partial^2 f_"}{\partial^2 \theta_2} =  \cos\theta_2(1-\eta_2-\kappa(1-\cos\theta_1)).
\end{align*}
 The Hessians are thus:

 $$
 H_1 = \begin{bmatrix}
\frac{\partial^2 f_1}{\partial^2 \theta_1} & \frac{\partial^2 f_1}{\partial \theta_1 \partial \theta_2}
\\ \frac{\partial^2 f_1}{\partial \theta_1 \partial \theta_2} & \frac{\partial^2 f_1}{\partial^2 \theta_2} 
 \end{bmatrix}
 =
\begin{bmatrix}
\cos\theta_1(1-\eta_1-\kappa(1-\cos\theta_2)) &  -\sin\theta_1\sin\theta_2 \\
 -\sin\theta_1\sin\theta_2 &\kappa\cos\theta_2(1-\cos\theta_1) 
 \end{bmatrix},$$ so $H_1|_{(\theta_1,\theta_2)=(0,0),\eta_1=\eta_2=0}=
 \begin{bmatrix}
     1 & 0 \\ 0 & 2\kappa \\
 \end{bmatrix}.$
 Analogously, we get to
 $$
 H_2 = \begin{bmatrix}
\frac{\partial^2 f_1}{\partial^2 \theta_1} & \frac{\partial^2 f_2}{\partial \theta_1 \partial \theta_2}
\\ \frac{\partial^2 f_2}{\partial \theta_1 \partial \theta_2} & \frac{\partial^2 f_2}{\partial^2 \theta_2} 
 \end{bmatrix}
 =
\begin{bmatrix}
\kappa\cos\theta_1(1-\cos\theta_2) &  -\sin\theta_1\sin\theta_2 \\
 -\sin\theta_1\sin\theta_2 &\cos\theta_2(1-\eta_2-\kappa(1-\cos\theta_1)) 
 \end{bmatrix},$$ so $H_2|_{(\theta_1,\theta_2)=(0,0),\eta_1=\eta_2=0} =
 \begin{bmatrix}
     2\kappa & 0 \\ 0 & 1 \\
 \end{bmatrix}.$ 
\\ Now let $u = \begin{pmatrix}a\\b \end{pmatrix}\in \ker(D_{\theta} f((0,0),(0,0,\kappa)))$, then $u^TH_1u=(a\ b)\left(\begin{pmatrix}
     1 & 0 \\ 0 & 2\kappa \\
 \end{pmatrix}\begin{pmatrix}
     a\\b
 \end{pmatrix}\right)=a^2+2\kappa b^2$ and
 $u^TH_2u=(a\ b)\left(\begin{bmatrix}
     2\kappa & 0 \\ 0 & 1 \\
 \end{bmatrix}\begin{pmatrix}
     a\\b
 \end{pmatrix}\right)=2\kappa a^2+b^2$. Hence $D_\theta^2f((0,0),(0,0,\kappa))(u,u)=\begin{pmatrix}
     a^2+2\kappa b^2\\2\kappa a^2+b^2
 \end{pmatrix}.$

Suppose $w_iD_\theta^2f((0,0),(0,0,\kappa))(v,v)= c(a^2+2\kappa b^2)+d(2\kappa a^2+b^2)=0$
for all non-zero $w_i = (c \ d)$, which is only possible if $a^2+2\kappa b^2=0$ and $2\kappa a^2+b^2=0$. Since $\kappa>0$, it follows that $a=b=0$, so $v=\begin{pmatrix} 0\\0 \end{pmatrix}$. Thus if $v = \begin{pmatrix}a\\b \end{pmatrix}\in \ker(D_{\theta} f((0,0),(0,0,\kappa)))$, by contraposition it follows that if $v \neq \begin{pmatrix} 0\\0\end{pmatrix}$, then $w_iD_\theta^2f((0,0),(0,0,\kappa))(v,v)\neq 0$ for at least one $1\leq i\leq j$, satisfying (SN3).\\

 Finally, since we showed that $D_{\theta} f((0,0),(0,0,\kappa))$ has no other eigenvalues than $0$, it is $k=0$ (SN4), completing the proof.
\end{proof}

\subsection{Counter-example to framework for smooth systems}\label{subsec:smoothnessIssue}

Replacing the assumption that $f$ is analytic by the assumption that $f$ is smooth, one can construct a counter-example to our definitions due to Peter Ashwin (personal communication, 2026) as follows:

Consider a planar system whose Jacobian is given by

\begin{align*}
D_f(p) = \label{eq:2}
\begin{cases}
\begin{bmatrix}
q & a(p) \\
0 & r
\end{bmatrix}, & $ p> 0$ \\[6pt]
\begin{bmatrix}
q & 0 \\
0 & r
\end{bmatrix}, & $p=0$ \\
\begin{bmatrix}
q & 0 \\
a(-p)
& r
\end{bmatrix}, & $p<0$\\
\end{cases}
\end{align*}

with parameters $p, q$ and $r$ and where $a(p)$ goes smoothly to zero (e.g., $a(p)=exp(-1/p^2)$). Then for $\rho_{\textnormal{crit}}=(p_{\textnormal{crit}},q_{\textnormal{crit}},r_{\textnormal{crit}})=(0,0,0)$, we have a
has double zero eigenvalue with geometric and algebraic multiplicity two, satisfying (SN1). Since \begin{equation*}
\frac{\partial D_f}{\partial q}\bigg|_{\rho_\textrm{crit}} = 
\begin{bmatrix}1&0\\0&0\end{bmatrix}, \qquad 
\frac{\partial D_f}{\partial r}\bigg|_{\rho_\textrm{crit}} = 
\begin{bmatrix}0&0\\0&1\end{bmatrix},
\end{equation*}
so choosing $w_1 = (1,0)$ and $w_2=(0,1)$ as left eigenvectors, we get 
$w_1(\partial f/\partial q)\neq 0$ and $w_2(\partial f/\partial r)\neq 0$, 
satisfying (SN2). However, for $p\neq0$, there is only a unique eigenvector which discontinuously flips from $\begin{pmatrix} 1\\0\end{pmatrix}$ to $\begin{pmatrix} 0\\1\end{pmatrix}$ on $p$ passing through $0$, hence violating conjecture \ref{eigvalConjecture}. While we do not construct an explicit ODE realizing $D_f(p)$ as a Jacobian at a ghost, the example suffices to illustrate the mechanism by which $C^\infty$ but non-analytic perturbations can violate the conjecture. However, it is likely that our definitions can be adapted for smooth systems to avoid such issues by assuming that eigenvectors vary continuously with parameters (albeit at the cost of having to check another condition).

\section{ Supplementary Figures}
\renewcommand{\figurename}{Supplementary Fig.}
\counterwithin{figure}{section}
\renewcommand\thefigure{\arabic{figure}}
\setcounter{figure}{0} 

\begin{figure}[h!]
	\centering
	\includegraphics[width=\textwidth]{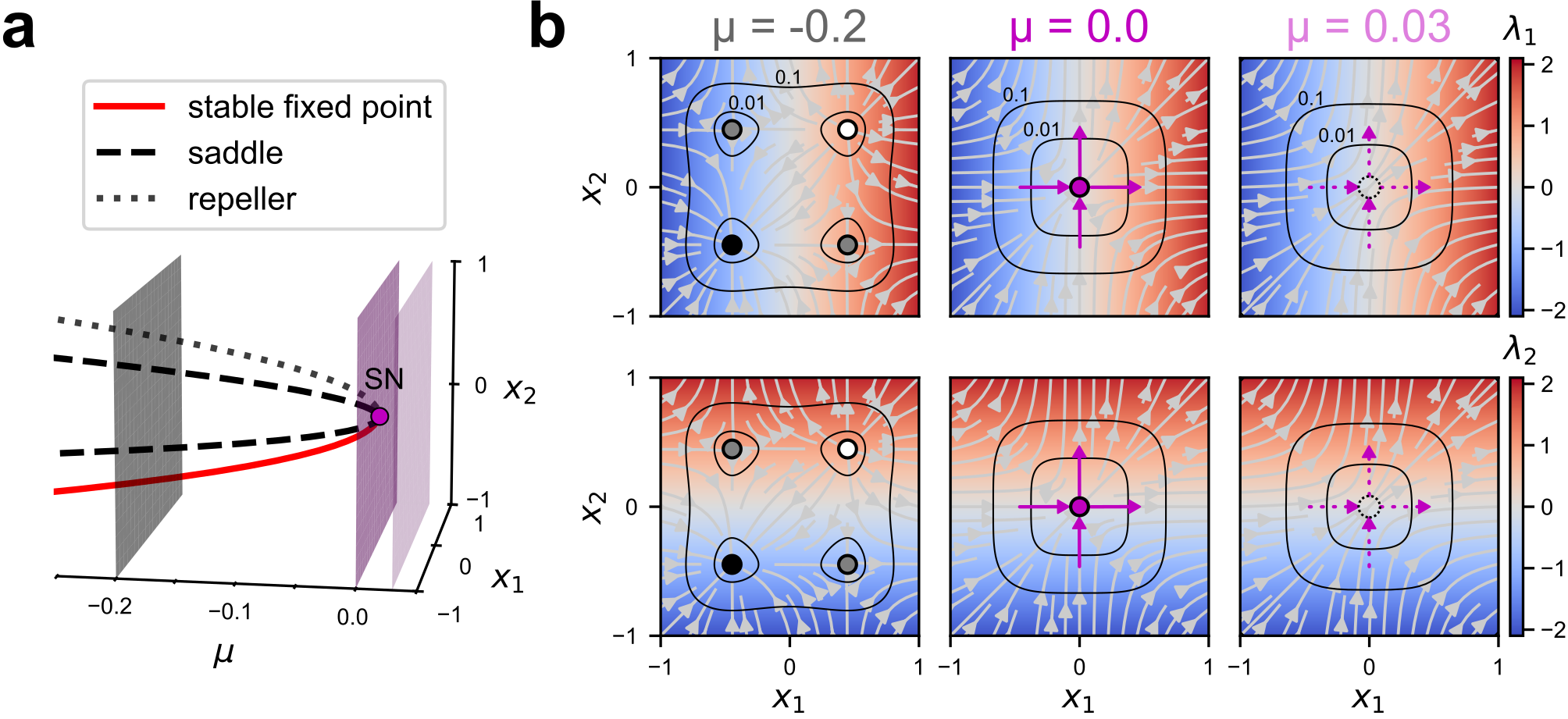}
	\caption{Local eigenvalue characteristics for fixed points and ghost of system \ref{eq:2}. \textbf{a}, saddle-node bifurcation in which four fixed points coalesce and form a type $2,0$ saddle-node. \textbf{b}, phase portrait and instantaneous eigenvalue distribution across space for different cross sections of the bifurcation diagram shown in \textbf{a}. black, grey, white and magenta dots indicate sinks, saddles and saddle-nodes, whereas dotted circles indicate ghosts. Grey arrows indicate flow. Black contour lines show different $Q$ levels, indicating regions of slow dynamics.}
	\label{sfig:2dghost}
\end{figure}

\begin{figure}
	\centering
	\includegraphics[width=\textwidth]{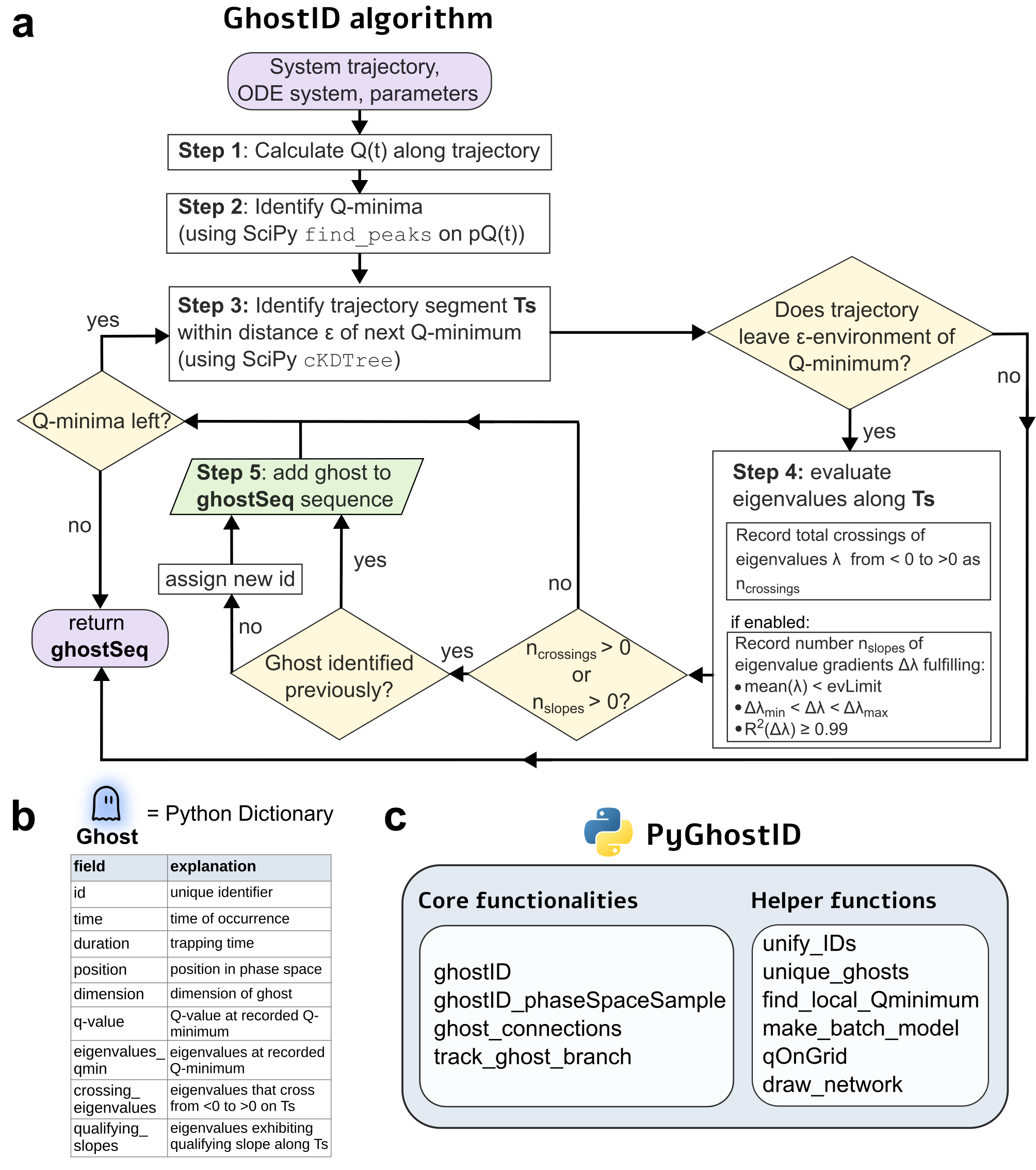}
	\caption{Algorithmic identification and characterization of ghosts. \textbf{a}, flow chart of the GhostID algorithm. \textbf{b}, information about identified ghosts stored by GhostID. \textbf{c}, functionalities offered to the user by the PyGhostID package (see text and package documentation for more details).}
	\label{sfig:algorithm}
\end{figure}

\begin{figure}[h]
\begin{center}
   \includegraphics[width=1\textwidth]{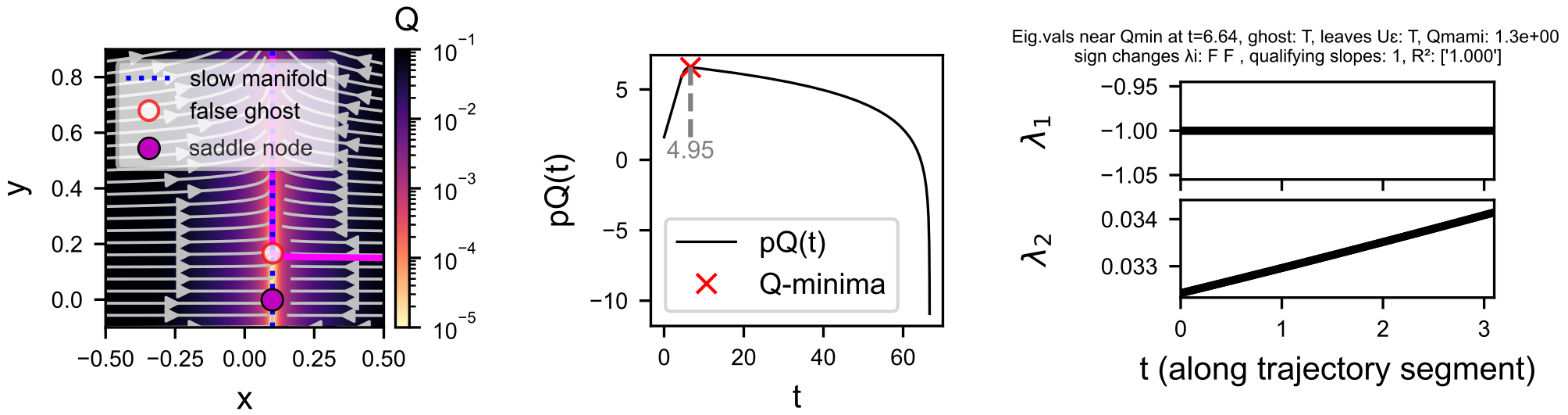}
\end{center}
\caption{False positive ghost identification with indirect method of GhostID in slow-fast toy model with a saddle-node from Kuehn \cite{Kuehn_2015} with parameter $\epsilon = 0.1$. Indirect identification was enabled via setting {\tt evLimit = 0.05} when calling PyGhostID's {\tt ghostID} function. Left: phase space. Middle: $pQ(t)$ plots along with $Q$-minima and peak prominences. Right: eigenvalues along trajectory segments and evaluation criteria.}
\label{sfig:falseGhost}
\end{figure}

\begin{figure}[h]
\begin{center}
   \includegraphics[width=1\textwidth]{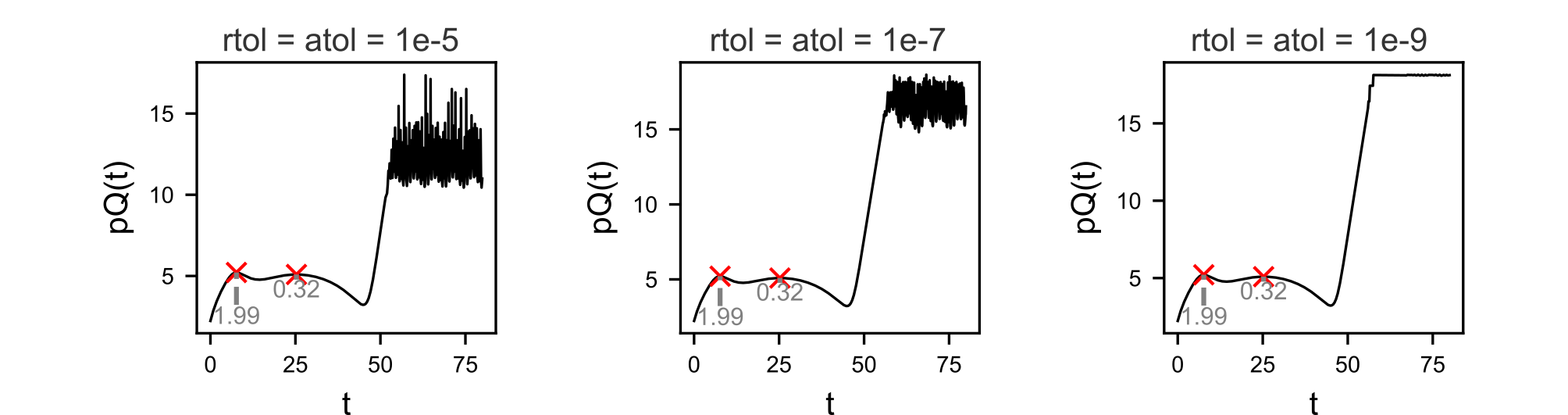}
\end{center}
\caption{Numerical noise in $Q$-values around fixed point in model from \cite{Bieg_2024}. Decreasing numerical tolerances when integrating the model before applying GhostID reduces irregularity of $pQ(t)$-timeseries.}
\label{sfig:numNoise}
\end{figure}

\begin{figure}[h]
\begin{center}
   \includegraphics[width=1\textwidth]{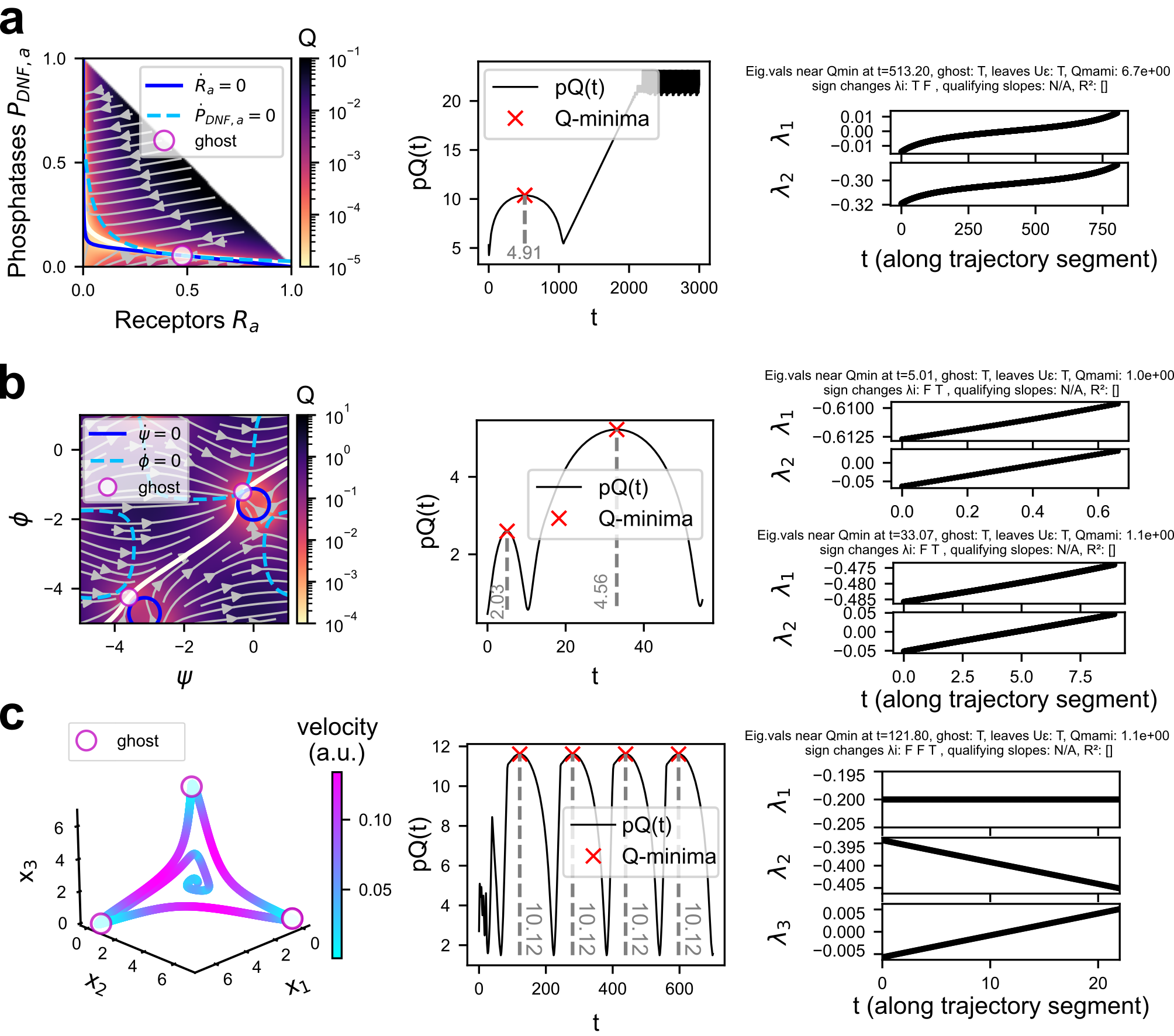}
\end{center}
\caption{Numerical validation of GhostID. \textbf{a}, ghost identified by GhostID in EGF-receptor model from \cite{Stanoev2018}. Parameter values: $LR_a=0$, $\alpha_1=0.0017$, $\alpha_2=0.3$, $\alpha_3=1$, $\hat{\gamma}_{DNF}=2.957$, $\hat{\beta}_{DNF}=36.0558$, $k_R=0.8$, $k_1=0.01$, $k_{2/1}=0.5$. \textbf{b}, ghosts identified by GhostID in charge-density wave model from \cite{Strogatz1989}. Parameter values: $E = 0.9$, $\gamma = 0.5$, $K=0.9$. \textbf{c}, ghosts identified by GhostID in gene-regulatory network model from \cite{Farjami_2021}. Parameter values: $b_1=b_2=b_3=10^{-5}$, $\alpha_1=\alpha_2=\alpha_3=9$,$\beta_1=\beta_2=\beta_3=0.1$, $h=3$, $d_1=d_2=d_3=0.2$. \textbf{a-c}, Left: phase space. Middle: $pQ(t)$ plots along with $Q$-minima and peak prominences. Right: eigenvalues along trajectory segments and evaluation criteria.}
\label{sfig:numval_ghosts}
\end{figure}

\begin{figure}[h]
\begin{center}
   \includegraphics[width=1\textwidth]{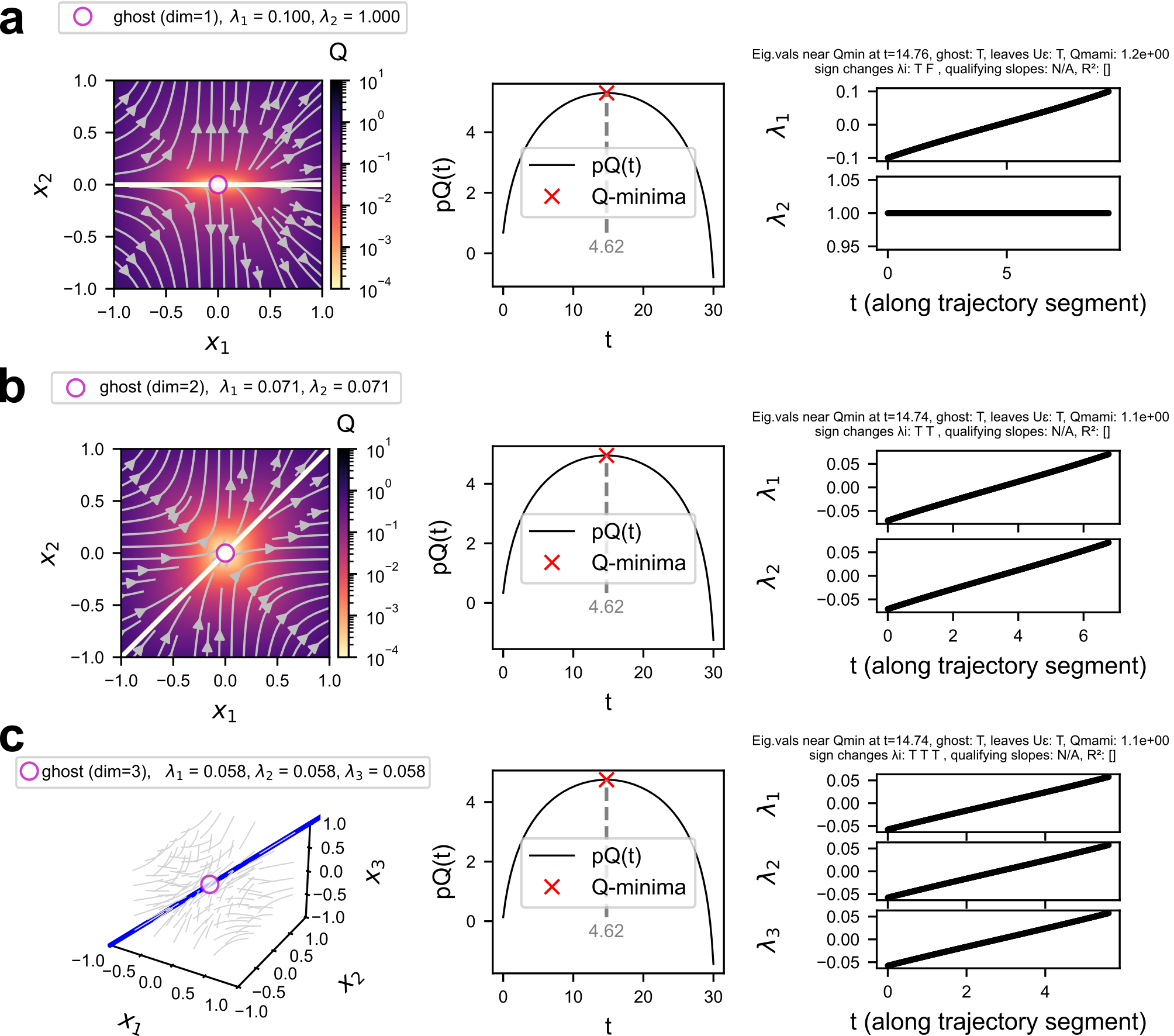}
\end{center}
\caption{Numerical validation of GhostID. New ghost types identified by GhostID in system 2. \textbf{a}, non-attracting ghost of a type $1,1$ saddle-node. Parameter values: $n = 2$, $j=1$, $k=1$, $\mu=0.01$. \textbf{b}, attracting ghost of a type $2,0$ saddle-node. Parameter values: $n = 2$, $j=2$, $k=0$, $\mu=0.01$. \textbf{c}, attracting ghost of a type $3,0$ saddle-node. Parameter values: $n = 3$, $j=3$, $k=0$, $\mu=0.01$.
\textbf{a-c}, Left: phase space. Middle: $pQ(t)$ plots along with $Q$-minima and peak prominences. Right: eigenvalues along trajectory segments and evaluation criteria.}
\label{sfig:numval_ghosts_new}
\end{figure}

\begin{figure}[h]
\begin{center}
   \includegraphics[width=1\textwidth]{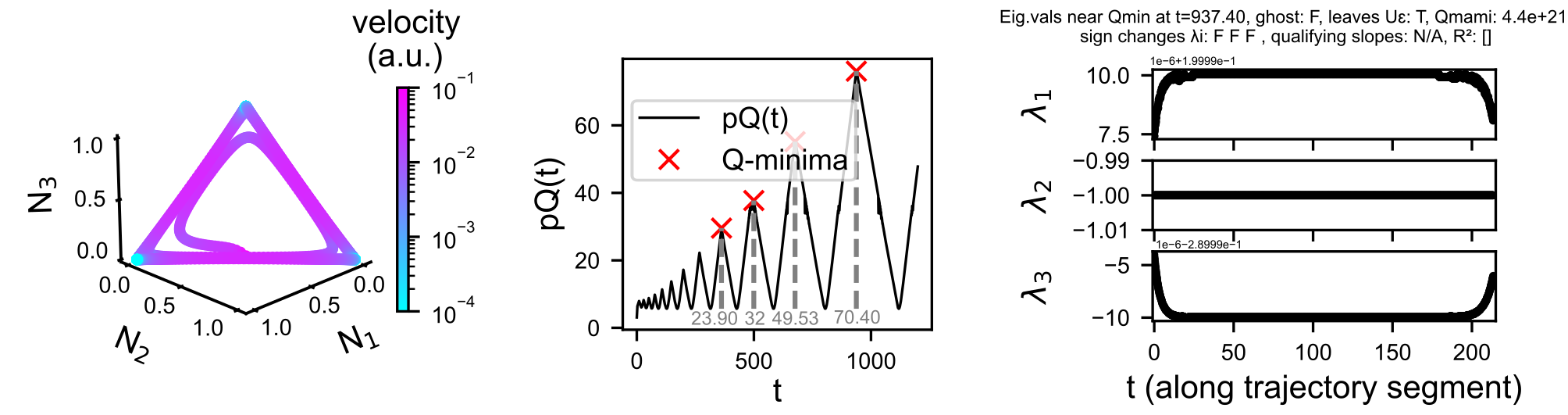}
\end{center}
\caption{Numerical validation of GhostID. 
Long transients in heteroclinic cycle from \cite{May1975} are not identified by GhostID. Parameter values: $\alpha=0.8$, $\beta=1.29$. Left: phase space. Middle: $pQ(t)$ plots along with $Q$-minima and peak prominences. Right: eigenvalues along trajectory segments and evaluation criteria.}
\label{sfig:numval_HC}
\end{figure}

\begin{figure}[h]
\begin{center}
   \includegraphics[width=1\textwidth]{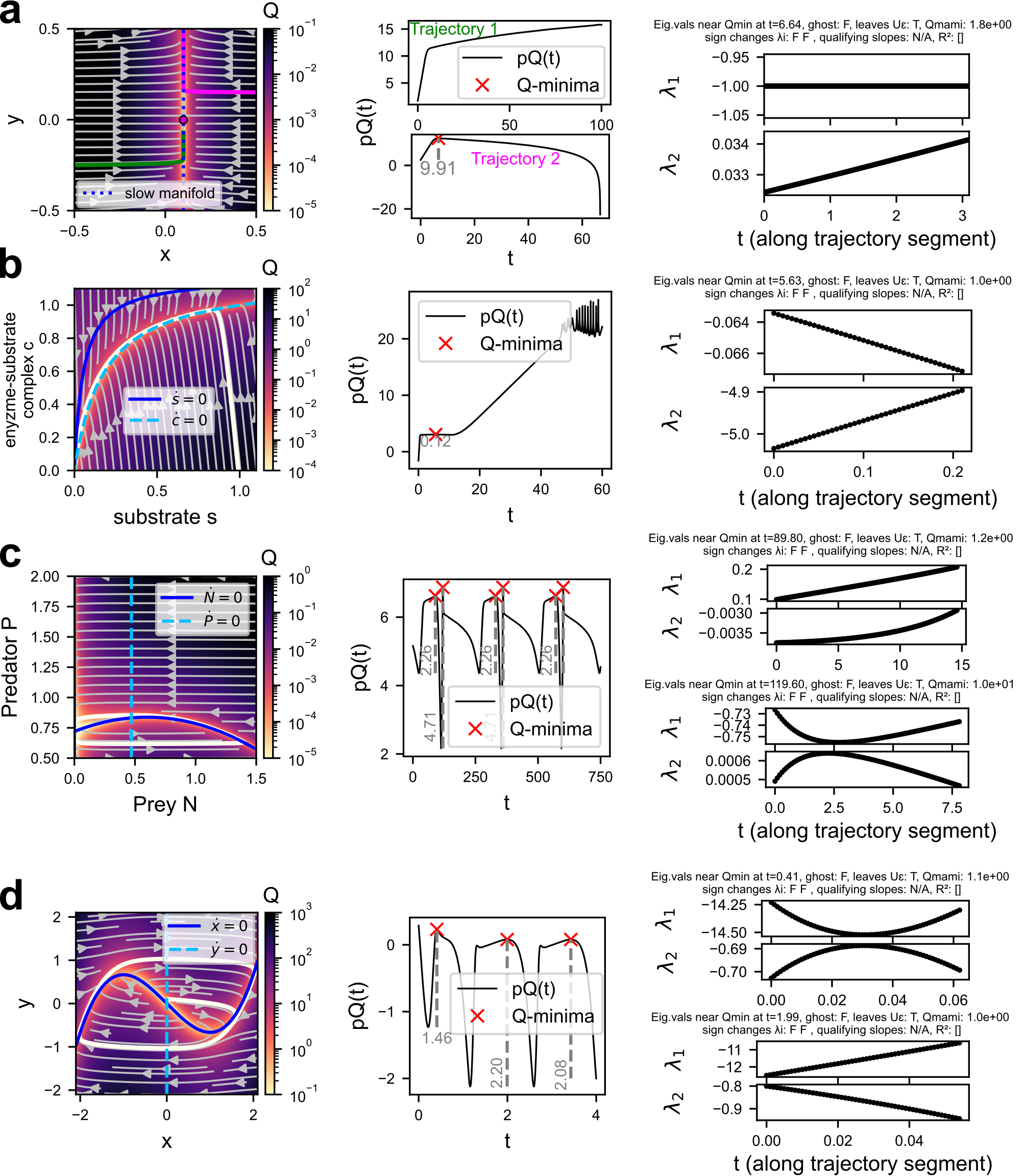}
\end{center}
\caption{Numerical validation of GhostID. 
Long transients in slow-fast systems are not identified by GhostID.
\textbf{a}, slow-fast toy model from \cite{Kuehn_2015}. Parameter values: $\epsilon = 0.1$.
\textbf{b}, Michaelis-Menten model of enzymatic catalysis from \cite{roussel2005singular}. Parameter values: $\alpha=0.85$, $\beta=0.3$, $\mu = 0.1$.
\textbf{c}, ecological model from \cite{Hastings_2018} in slow-fast regime. Parameter values: $\gamma=2.5$, $h=1$, $v=0.5$, $m=0.4$, $\alpha=1.8$, $K=2.2$, $\epsilon=0.01$
\textbf{d}, van der Pol oscillator. Parameter values: $\epsilon = 0.1$.
\textbf{a-d}, Left: phase space. Middle: $pQ(t)$ plots along with $Q$-minima and peak prominences. Right: eigenvalues along trajectory segments and evaluation criteria.}
\label{sfig:numval_SF}
\end{figure}

\begin{figure}[h]
\begin{center}
   \includegraphics[width=1\textwidth]{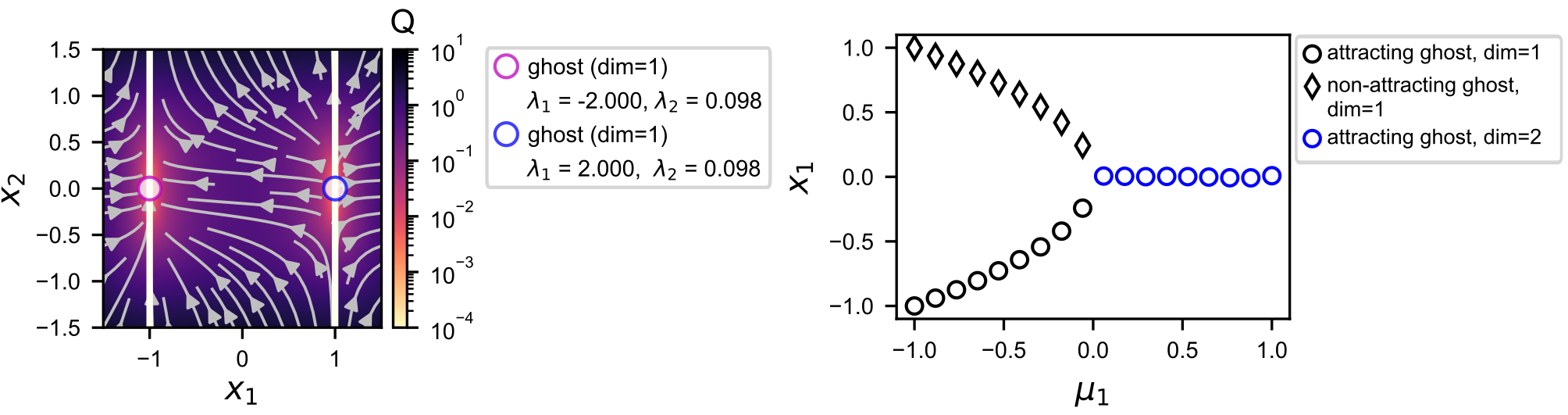}
\end{center}
\caption{Ghost saddle-node bifurcation in system (2) for. Left: phase space at parameter values $\mu_1=-1, \mu_2=0.05$. GhostID identifies two one-dimensional ghosts from the trajectories shown in white, one attracting (magenta), the other non-attracting (blue). Right: Tracking these ghosts for increasing values of $\mu_1$ shows how attracting and non-attracting ghosts come closer and closer together as $\mu_1$ increases until they merge at $\mu_1=0$ and form a two-dimensional ghost for $\mu_1>0$. We term this new type of bifurcation a ghost saddle-node bifurcation.}
\label{sfig:gSN_sys2}
\end{figure}

\FloatBarrier
\section{Use of generative AI}\label{subsec:genAIuse}

ChatGPT (GPT-5) has been used for specific coding tasks:
\begin{itemize}
    \item After a first version of GhostID was coded by D.K., GPT-5 was instructed to evaluate numerical performance, after which the following suggestions were implemented to speed up the algorithm: using JAX for vectorizing and batch-processing calculations of $pQ(t)$, finding points in a neighborhood of a given point using SciPy's {\tt cKDTree} function, fast function for linear regression used in indirect method of eigenvalue evaluation.
    \item The following functions of PyGhostID were written by GPT-5 after being given precise instructions of what the input data and the parameters are, how the output should look like, and what steps/calculations should be done for processing the input: 
    {\tt ghostID\textunderscore phaseSpaceSample, make\textunderscore batch\textunderscore model, find\textunderscore local\textunderscore Qminimum,qOnGrid,iqr\textunderscore sliding\textunderscore filter,make\textunderscore jacfun,slope\textunderscore and\textunderscore r2}
    \item Generating comments for functions
    \item Debugging
\end{itemize}

In addition, after first formulation of definitions \ref{def:SN} and \ref{def:ghost} as well as theorem \ref{theorem_vicinityHypFP} by D.K., ChatGPT (GPT-5.1) and Anthropic's Claude (Sonnet 4.6) were used to evaluate their plausibility and consistency, which led to the following changes:
\begin{itemize}
    \item Definition \ref{def:SN}: 
    \begin{itemize}
        \item (SN1) requiring eigenvalues to also have a geometric multiplicity of $j$, to distinguish from other bifurcations (e.g. Bogdanov-Takens) and ensure consistency with (SN2) and (SN3).
        \item (SN3) changed `$w_i(D_x^2f(x^*,\rho_{\textnormal{crit}})(v_i,v_i))\neq 0$ for all $1\leq i \leq j$' to `For every non-zero $v\in \ker(D_xf(x^*,\rho_{\textnormal{crit}}))$, there exists $i$ with $1\leq i \leq j$ such that $w_i(D_x^2f(x^*,\rho_{\textnormal{crit}})(v,v))\neq 0$' to make the criterion independent of the particular choice of base.
    \end{itemize}
\end{itemize}

All AI outputs have been checked carefully by the authors who take full responsibility for the code and the results presented in this study. AI has \underline{not} been used in any of the following areas of this work: Idea generation and conceptualization, numerical simulations and data analysis, writing of the manuscript, creation of figures, creation and formatting of bibliography.



\section{Supplementary References}
\printbibliography[heading=none]
\end{refsection}

\end{document}